\documentclass[a4paper, 10pt]{article}
\usepackage[margin=.95in,dvips]{geometry}
\usepackage[bottom]{footmisc}
\usepackage[utf8]{inputenc}
\usepackage{lmodern}
\usepackage{indentfirst}
\usepackage[affil-it]{authblk}
\usepackage{rotating}
\usepackage{tocloft}
\usepackage{titlesec}
 \titleformat{\subparagraph}[runin]{\normalfont}{\thesubparagraph}{0pt}{\underline}[.]
 \titleformat{\paragraph}[hang]{\normalfont\bfseries}{\theparagraph}{0pt}{}
\usepackage{titletoc}

\usepackage[textsize=scriptsize]{todonotes} 
\makeatletter
\renewcommand{\todo}[2][]{\tikzexternaldisable\@todo[#1]{#2}\tikzexternalenable}
\makeatother

\usepackage{standalone}
\usepackage{subfiles}

\usepackage[hidelinks,bookmarksopen,bookmarksdepth=2]{hyperref}
\usepackage{enumitem}
\setlist[enumerate]{itemsep=2.0pt plus 1.0 pt minus 0.5pt, topsep=4.0pt plus 2.0 pt minus 1.0pt}
\setlist[itemize]{itemsep=2.0pt plus 1.0 pt minus 0.5pt, topsep=4.0pt plus 2.0 pt minus 1.0pt}
\usepackage{tikz}
\usetikzlibrary{external}
\usepackage{calc}
\usepackage{graphicx}
\usepackage{subcaption}
\usepackage{epsfig}
\usepackage{color, soul}

\usepackage{marginnote}

\usepackage{color,soulutf8}

\usepackage{amsmath}
\usepackage{amsthm, slashed}
\usepackage[capitalize]{cleveref}
\usepackage{amssymb}
\usepackage{array}
\usepackage{xfrac}
\usepackage{bbm}
\usepackage{bm}
\usepackage{marvosym}
\usepackage{mathtools}
\usepackage{mathrsfs}
\usepackage{upgreek}
\usepackage{accents}
\usepackage{textcomp}
\usepackage{slashed} %for gradient on spheres
\usepackage{harmony} %for musical symbols
\usepackage{wasysym} %for clock
\wasyfamily
\DeclareFontShape{U}{wasy}{b}{n}{ <-10> ssub * wasy/m/n
 <10> <10.95> <12> <14.4> <17.28> <20.74> <24.88>wasyb10 }{}
\DeclareMathAlphabet\mathbfcal{OMS}{cmsy}{b}{n}
\newcommand\numberthis{\addtocounter{equation}{1}\tag{\theequation}}

\allowdisplaybreaks

\numberwithin{equation}{section}

\newtheorem{definition}{Definition}[section]

\newtheorem*{theorem*}{Theorem}
 
\newtheorem*{conjecture*}{Conjecture}

\newtheorem*{corollary*}{Corollary}

\newtheorem{proposition}{Proposition}[section]
\newtheorem{lemma}[proposition]{Lemma}
\newtheorem{remark}[proposition]{Remark}

\newtheoremstyle{mystyle}%                % Name
  {}%                                     % Space above
  {}%                                     % Space below
  {\itshape}%                                     % Body font
  {}%                                     % Indent amount
  {\bfseries}%                            % Theorem head font
  {.}%                                    % Punctuation after theorem head
  { }%                                    % Space after theorem head, ' ', or \newline
  {\thmname{#1}\thmnumber{ #2}\thmnote{ (#3)}}%                                     % Theorem head spec (can be left empty, meaning `normal')

\theoremstyle{mystyle}
\newtheorem{bigfact}{Fact}

\usepackage{csquotes}
\usepackage[backend=biber,
style=alphabetic, date=year, doi=false,isbn=false,maxnames=10]{biblatex}
\renewbibmacro{in:}{}

\AtEveryBibitem{\clearfield{pages}} 
\DeclareBibliographyCategory{needsurl}

\renewbibmacro*{url+urldate}{%
  \ifcategory{needsurl}
    {\printfield{url}%
     \iffieldundef{urlyear}
       {}
       {\setunit*{\addspace}%
        \printurldate}}
    {}}

\newcommand{\mb}[1]{\mathbb{#1}}

\newcommand{\mc}[1]{\mathcal{#1}}
\newcommand{\mr}[1]{\mathrm{#1}}

\newcommand{\p}{\partial}
\newcommand{\lp}{\left}
\newcommand{\rp}{\right}

\renewcommand{\d}{{\rm d}}

\newcommand{\swei}[2]{#1^{[{#2}]}}

\newcommand{\sml}[2]{#1_{ml}^{[{#2}],\,a\omega}}
\newcommand{\smlambda}[2]{#1_{m\Lambda}^{[{#2}],\,a\omega}}

\newcommand{\bmLslash}{\textup{\textbf{\L{}}}}
\newcommand{\Lslash}{\textup{{Ł}}}

\usepackage{comment}

\title{\textbf{The Teukolsky--Starobinsky constants:}\\ \textbf{facts and fictions}}

\author[1,2]{{\Large Marc Casals}}
\author[3]{{\Large  Rita \mbox{Teixeira da Costa}\vspace{0.4cm}}}

\affil[1]{\small  Centro Brasileiro de Pesquisas Físicas (CBPF),
Rio de Janeiro,
CEP 22290-180, Brazil  }
\affil[2]{School of Mathematics and Statistics, University College Dublin, Belfield, Dublin 4, Ireland \vspace{0.2cm}\ }

\affil[3]{\small 
University of Cambridge, Center for Mathematical Sciences, Cambridge CB3 0WA, United Kingdom}

\date{\today}

\begin{document}

\maketitle

\vspace{-0.5cm}

\begin{abstract}
The Teukolsky Master Equation describes the dynamics of massless fields with spin on a Kerr black hole. Under separation of variables, spin-reversal for this equation is accomplished through the so-called Teukolsky--Starobinsky identities. These identities are associated to  the so-called Teukolsky--Starobinsky constants, which are spin-dependent.

We collect some properties of the Teukolsky--Starobinsky constants and dispel some myths present in the literature. We show that, contrary to popular belief,  these constants can be negative for spin larger than 2. Such fields thus exhibit a novel form of energy amplification which occurs for non-superradiant frequencies.
\end{abstract}

\section{Introduction}

Since as early as the 1950s \cite{Regge1957}, the study of perturbations of stationary black holes has been a central theme of research in classical General Relativity. In 4 spacetime dimensions and in vacuum, the paradigmatic and, conjecturally, only examples of stationary black hole solutions  are the rotating black holes in the Kerr family \cite{Kerr1963}, parametrized by mass $M>0$ and specific angular momentum $|a|\leq M$.

While studying stability of Kerr, Teukolsky \cite{Teukolsky1973} introduced what became to be known as the Teukolsky Master Equations. In Boyer--Lindquist coordinates $(t,r,\theta,\phi)\in\mathbb{R}\times(r_+\equiv M+\sqrt{M^2-a^2},\infty)\times\mathbb{S}^2$, these are given by
\begin{align}
\begin{split}
\bigg[\Box_{g_{a,M}} &+\frac{2(\pm s)}{\rho^2}(r-M)\p_r +\frac{2(\pm s)}{\rho^2}\lp(\frac{a(r-M)}{\Delta}+i\frac{\cos\theta}{\sin^2\theta}\rp)\p_\phi\\
&+\frac{2(\pm s)}{\rho^2}\lp(\frac{M(r^2-a^2)}{\Delta}-r -ia\cos\theta\rp)\p_t +\frac{1}{\rho^2}\lp(\pm s-s^2\cot^2\theta\rp)\bigg]\upalpha^{[\pm s]}  =0\,.
\end{split}\label{eq:Teukolsky-equation-intro}
\end{align}
where  $\Delta=r^2-2Mr+a^2$, $\rho^2=r^2+a^2\cos^2\theta$ and $\Box_{g_{a,M}}$ the covariant scalar wave operator on the Kerr background. Here, $\pm s$ denotes a spin parameter, making $\upalpha^{[\pm s]}$ a spin-weighted function; this notion, and hence \eqref{eq:Teukolsky-equation-intro}, is well-defined for any $s\in\frac12\mathbb{Z}_{\geq 0}$, see \cite{Penrose1984}. The Teukolsky Master Equations~\eqref{eq:Teukolsky-equation-intro} play a crucial role in characterizing perturbations of Kerr black holes. For $s=0$, they reduce to the scalar wave equation on Kerr. For $s=1$ and $2$, they describe the dynamics of gauge-invariant electromagnetic and curvature quantities under the linearized Maxwell and Einstein equations, respectively, in the Newman--Penrose formalism \cite{Newman1962}. Some half-integer spin cases can also be  interpreted physically: $s=1/2$ corresponds to the Dirac equation for (massless) neutrinos; the case $s=3/2$ is known as the Rarita--Schwinger equation. 
The  case $s>2$ of the Teukolsky Master Equations \eqref{eq:Teukolsky-equation-intro}, although less often considered, 
is of mathematical interest in its own right and  might play a role in higher-spin theories \cite{Vasiliev1996}.

In this paper, we focus on separable solutions of \eqref{eq:Teukolsky-equation-intro} (see \cite{Carter1968,Teukolsky1973}): for $\omega\in\mathbb{R}$, $m-s\in\mathbb{Z}$ and $l\in\mathbb{Z}_{\geq\max\{|m|,s\}}$, these take the form
\begin{align*}
\swei{\upalpha}{\pm s}(t,r,\theta,\phi)=e^{-i\omega t}e^{im\phi}\sml{S}{\pm s}(\theta)\sml{\upalpha}{\pm s}(r)\,, 
\end{align*}
where $\sml{S}{\pm s }$ and $\sml{\upalpha}{\pm s}$ satisfy ODEs,  respectively
referred to as the angular ODE, given below as \eqref{eq:angular-ode}, and the radial ODE, given below as \eqref{eq:radial-ODE-alpha}.

As discovered by Starobinsky \cite{Starobinsky1974} and Teukolsky \cite{Teukolsky1974} for $s=1,2$ and later generalized to all $s\in\frac12\mathbb{Z}_{\geq 0}$ \cite{Kalnins1989}, the angular and radial Teukolsky ODEs of spin $\pm  s$ exhibit a curious property: on applying a certain first order differential operator $s$ times to a solution to the ODE with spin $+ s$, one obtains a solution to the ODE with spin $-s$, and vice-versa. Though seldom remarked in the classical literature, each of the radial Teukolsky--Starobinsky operators of spin $\pm s$ acts differently on the ingoing and outgoing components of the radial solutions, and a similar statement is true in the angular setting,  see Propositions~\ref{prop:TS-angular} and \ref{prop:TS-radial} below. Nevertheless, by applying these operators in succession to a solution of the angular or radial ODE, one can check that, at least for $|s|\leq 3$, the exact same solution is recovered up to a constant: respectively, the \textit{angular} Teukolsky--Starobinsky constant, denoted by $\mathfrak{B}_s=\mathfrak{B}_s(a\omega,m,l)$, and the \textit{radial} Teukolsky--Starobinsky constant, denoted by 
$$\mathfrak{C}_s=\mathfrak{C}_s(a,M,\omega,m,l)\,.$$
It is on the radial Teukolsky--Starobinsky constant, $\mathfrak{C}_s$, that we focus on for the rest of this introduction.

In the classical literature, see e.g.\ \cite{Kalnins1989,Chandrasekhar1990,Kalnins1992}, $\mathfrak{C}_s$ is frequently denoted as the square of a complex number. Taking this claim literally, we would conclude
\begin{align}
s\in\frac12\mathbb{Z}_{\geq 0}\implies \mathfrak{C}_s(a,M,\omega,m,l)\geq 0\,\,\,\, \forall\, (\omega,m,l) \text{~real~}\,. \label{eq:bogus-nonnegativity-claim}\tag{$\times$}
\end{align}
To justify \eqref{eq:bogus-nonnegativity-claim}, some authors point to the fact that $\swei{\upalpha}{-s}$ and $\Delta^s\swei{\upalpha}{+s}$ satisfy complex conjugate equations. Unfortunately, this argument is incorrect: \eqref{eq:bogus-nonnegativity-claim} does not follow from the fact that these radial ODEs are complex conjugates of each other. In fact, we show that \eqref{eq:bogus-nonnegativity-claim} is manifestly \textit{false} for general spin $s$:

\begin{bigfact}[TS constant sign I] \label{fact:negativity} The radial Teukolsky--Starobinsky constant can be negative: for any $|a|\in(0,M]$,
$$s\in\lp\{\frac52,3\rp\} \implies \exists\, (\omega,m,l) \text{~real such that~~}\mathfrak{C}_s(a,M,\omega,m,l)<0 \,.$$
\end{bigfact}
\noindent The proof of Fact~\ref{fact:negativity} is short and elementary: we require only the high frequency expansions for spin-weighted spheroidal angular eigenvalues, on which $\mathfrak{C}_s$ depends  (see Lemma~\ref{lemma:TS-radial-constant-examples}), which were obtained\footnote{Earlier work on these limits \cite{Breuer1977} suffered from flaws which were corrected in the references given.} in recent work of the first author and collaborators \cite{Casals2005,Casals2018}.  We note here the importance of the assumption $a\neq 0$: it is well-known that $\mathfrak{C}_s(a=0,M,\omega,m,l)\geq 1$, at least if $|s|\leq 3$.

The implications of Fact~\ref{fact:negativity} are quite surprising. To explain these, recall that the Teukolsky--Starobinsky identities serve to define an energy identity for $\swei{\upalpha}{\pm s}$, which will thus depend on the constants $\mathfrak{C}_s$; this observation goes back to \cite{Teukolsky1974} but the reader may find a rigorous statement below in Lemma~\ref{lemma:TS-energy-identity}. In the case $s=0$, corresponding to the scalar wave equation, where $\mathfrak{C}_{s=0}\equiv 1$ plays no role in the energy, it is well-known that \textit{energy amplification} occurs if and only if the frequency parameters $(\omega,m,l)$ are superradiant, i.e.\ such that
\begin{align}
 \omega(\omega-m\upomega_+)<0\,, \qquad \upomega_+\equiv \frac{a}{2Mr_+}\,. \label{eq:superrad-frequencies-intro}
\end{align}
Condition \eqref{eq:superrad-frequencies-intro} is intimately tied to the Kerr geometry, as $\upomega_+$ is uniquely specified by the Kerr parameters and the amplification effect generated can be linked to the presence of an ergoregion in Kerr. For higher $s$, the superradiant condition is
\begin{align}
s\in\mathbb{Z}_{\geq 0}\text{~~~and~\eqref{eq:superrad-frequencies-intro} holds} \,, \label{eq:superrad-frequencies-intro-higher-s}
\end{align}
as the half-integer spin particles do not interact with the ergorregion.   
If one could show $\mathfrak{C}_s\geq 0$ unconditionally, then superradiance \eqref{eq:superrad-frequencies-intro-higher-s} would be the only source of energy amplification, with half-integer spins experiencing none. Such claims are often made in the classical literature. However, in view of Fact~\ref{fact:negativity}, we see that this is another fiction: in general, energy amplification may occur for half-integer spins and, in general, it may occur for integer spins and $(\omega,m,l)$ not in \eqref{eq:superrad-frequencies-intro-higher-s}:
\begin{bigfact}[Non-superradiant amplification] \label{fact:new-superradiance} If $s=5/2$ and $|a|\in(0,M]$, there are real $(\omega,m,l)$ for which there is energy amplification. If $s=3$ and $|a|\in(0,M]$, there are real $(\omega,m,l)$ such that the superradiant condition does not hold, i.e.\  $\omega(\omega-m\upomega_+)>0$, but for which there is energy amplification.
\end{bigfact}
\noindent Some numerical evidence of this novel non-superradiant amplification effect is given below in Section~\ref{sec:nonsuperradiant-amplification}.  

\begin{remark}
It is important to note that the novel amplification effect uncovered in Fact~\ref{fact:new-superradiance} in no way invalidates the mode stability theorems for \eqref{eq:Teukolsky-equation-intro} obtained in \cite{Whiting1989,Shlapentokh-Rothman2015,Andersson2017,TeixeiradaCosta2019}. To explain why this is, we briefly review the strategy of these works.

For spin $-s$, mode stability is shown in these works by application of the transformations introduced in \cite{Whiting1989,TeixeiradaCosta2019}. These transformations map  solutions of \eqref{eq:Teukolsky-equation-intro} to solutions of a scalar wave equation  on a new, ergoregionless, spacetime. The energy identity for this new scalar wave equation with real potential, as one expects for scalar fields, does not depend on any Teukolsky--Starobinsky-type constant.

Now consider spin $+s$. If the frequency triple $(\omega,m,l)$ is such that the Teukolsky--Starobinsky $\mathfrak{C}_s(a,M,\omega,l)$ constant vanishes, then one may easily deduce that there is no mode solution associated to $(\omega,m,l)$, i.e.\ no separable solution to \eqref{eq:Teukolsky-equation-intro} which is outgoing at the spacetime's future null infinity and ingoing at the black hole's future event horizon, see \cite[Lemma 2.19]{TeixeiradaCosta2019}. Otherwise, if $(\omega,m,l)$ is such that $\mathfrak{C}_s(a,M,\omega,l)\neq 0$, the Teukolsky--Starobinsky identities are applied to show that mode stability for spin $-s$ implies mode stability for spin $+s$.
\end{remark}

Remarkably, even though the classical argument commonly used to justify this is  incorrect, non-negativity  of $\mathfrak{C}_s$ does hold for the physical spins $s\leq 2$:
\begin{bigfact}[TS constant sign II] \label{fact:positivity} If $s\leq 2$, the radial Teukolsky--Starobinsky constant is never negative: for any $|a|\in[0,M]$,
$$s\in\lp\{\frac12,1,\frac32,2\rp\}\implies\mathfrak{C}_s(a,M,\omega,m,l)\geq 0 \quad  \forall\, (\omega,m,l) \text{~real}\,.$$
Hence, for $s\in\{1/2,3/2\}$, there is no energy amplification, and for $s\in\{1,2\}$, energy amplification occurs if and only if \eqref{eq:superrad-frequencies-intro} holds.
\end{bigfact}
\noindent Fact~\ref{fact:positivity} follows not from any property of the radial ODE alone, but from a comparison between the two Teukolsky--Starobinsky constants $\mathfrak{C}_s$ and $\mathfrak{B}_s$, as the latter can be easily shown to have definite sign. To the best of our knowledge, it was Teukolsky and Press \cite{Teukolsky1974} who first noted this in the $s=1$ case.

\begin{remark} Fact~\ref{fact:positivity} implies that, as claimed in the literature, $\mathfrak{C}_s$ can \textbf{sometimes}---whenever $s\leq 2$, to be precise---be denoted as the square of a complex number. Most literature available, such as the classical reference \cite{Chandrasekhar}, focuses precisely on such cases. However, Fact~\ref{fact:negativity} highlights the importance of stressing the caveat in bold.
\end{remark}

In light of the results present here for the sign of $\mathfrak{C}_s$, it is natural to try to understand whether $\mathfrak{C}_s(a,M,\omega,m,l)=0$ for some real $(\omega,m,l)$. Such frequencies are known as real algebraically special \cite{Wald1973,Chandrasekhar1984}. Since $\mathfrak{C}_s(a,M,\omega=0,m,l)\geq 1$, at least for $s\leq 3$, it follows from Fact~\ref{fact:negativity} that
\begin{bigfact}[AS frequencies I] \label{fact:AS-frequencies-exist} There are real algebraically special frequencies for $s> 2$: for any $|a|\in(0,M]$, 
$$s\in\lp\{\frac52,3\rp\} \implies \exists\, (\omega,m,l) \text{~real such that~~}\mathfrak{C}_s(a,M,\omega,m,l)=0 \,.$$
\end{bigfact}

In the case $s\leq 2$, Fact~\ref{fact:positivity} gives one hope of ruling out algebraically special frequencies. We show that these hopes are well-founded, thus answering a question raised in \cite{Wald1973}:
\begin{bigfact}[AS frequencies II] \label{fact:AS-frequencies-dont-exist} There do not exist real algebraically special frequencies for $s\leq 2$: for any $|a|\in[0,M]$, 
$$s\in\lp\{\frac12,1,\frac32,2\rp\} \implies \mathfrak{C}_s(a,M,\omega,m,l)>0 \quad  \forall\, (\omega,m,l) \text{~real}\,.$$
\end{bigfact}
\noindent To the best of our knowledge, prior to this work, this  result had only been noted in the $s=2$ case, in work of the second author and collaborators, in \cite{TeixeiradaCosta2019,SRTdC2020}.

We contrast Fact~\ref{fact:AS-frequencies-dont-exist} with the following result
\begin{bigfact}[TS lower bound] \label{fact:lower-bounds} For any $a\in (0,M]$,  as $|\omega|\to \infty$, we have 
\begin{align*}
s\in\lp\{\frac12,1,\frac32\rp\} &\implies \exists\, (m,l) \text{~such that~~}\mathfrak{C}_s(a,M,\omega,m,l)=O(|\omega|^{-N})\,\,\,  \forall\,N>0\,,\text{~~hence~~} \inf\mathfrak{C}_s=0\,;\\
s=2 &\implies \forall\,(m,l)~~\mathfrak{C}_s(a,M,\omega,m,l)=O(\omega^{2})\,,  \text{~~hence~~} \inf \mathfrak{C}_2>0\,.
\end{align*}
\end{bigfact}
Indeed, for fixed $a\neq 0$, we show that the  limit when $|\omega|\to\infty$ is algebraically special for some $(l,m)$ if $s\leq 3/2$. In the case $s=2$, no such limit can be algebraically special: the comparison with the angular Teukolsky--Starobinsky constant gives $\mathfrak{C}_2(a,M,\omega,m,l)\geq 144M^2\omega^2$ and it is the latter term which ensures that  $\mathfrak{C}_2$ has a positive lower bound.

To conclude this introduction, we remark that we fully expect that our Facts~\ref{fact:negativity}, \ref{fact:new-superradiance} and \ref{fact:AS-frequencies-exist} can be generalized to higher $s\in\frac12 \mathbb{Z}_{\geq 0}$. Furthermore, we expect Facts~\ref{fact:positivity}, \ref{fact:AS-frequencies-dont-exist} and \ref{fact:lower-bounds} to also hold for the  Kerr--(anti-)de Sitter black hole spacetimes, for which there is  a generalization of the Teukolsky--Starobinsky identities and constants \cite{TorresdelCastillo1988} (see also the more recent \cite{Dias2013}).

\vspace{\baselineskip}

\noindent \textbf{Acknowledgments.}  M.C.\ acknowledges partial financial support by CNPq (Brazil), process number 310200/2017-2.
R.TdC.\ acknowledges support from EPSRC (United Kingdom), grant EP/L016516/1, and thanks Simon Becker for helpful suggestions concerning the analysis of the angular ODE in this paper and Yakov Shlapentokh-Rothman for numerous discussions.
This work makes use of the Black Hole Perturbation Toolkit.

\section{The angular Teukolsky--Starobinsky constants}
\label{sec:angular-TS-constants}

In this section, we introduce the angular ODE corresponding to the Teukolsky equation~\eqref{eq:Teukolsky-equation-intro}. We then define the angular Teukolsky--Starobinsky constants and characterize their sign.

\subsection{The angular ODE}
\label{sec:angular-ODE}

Consider the angular ODE
\begin{gather}
\begin{gathered}
-\frac{1}{\sin\theta}\frac{d}{d\theta}\lp(\sin\theta\frac{d}{d\theta}\rp)\Xi_{m\Lambda}^{[\pm s],\,(\nu)}(\theta)
+ \lp(\frac{(m\pm s\cos\theta)^2}{\sin^2\theta}+\nu^2\sin^2\theta+2\nu (\pm s) \cos\theta\rp)\Xi_{m\Lambda}^{[\pm s],\,(\nu)}(\theta)\\ =\Lambda\cdot  \Xi_{m\Lambda}^{[\pm s],\,(\nu)}(\theta) \,,
\end{gathered} \label{eq:angular-ode}
\end{gather}
where $\theta\in(0,\pi)$, $s\in\frac12\mathbb{Z}_{\geq 0}$, $m-s\in\mathbb{Z}$, $\nu\in\mathbb{R}$ and $\Lambda\in\mathbb{R}$. An asymptotic analysis leads us to conclude that for a solution to \eqref{eq:angular-ode} there is a unique set of real numbers $\swei{a}{\pm s}_i$, $i=1,\dots, 4$, such that
\begin{align}
\Xi_{m\Lambda}^{[\pm s],\,(\nu)} &=\swei{a}{\pm s}_1\swei{\Xi}{\pm s}_{{\rm norm},1}+\swei{a}{\pm s}_2\swei{\Xi}{\pm s}_{{\rm norm},2}=\swei{a}{\pm s}_3\swei{\Xi}{\pm s}_{{\rm norm},3}+\swei{a}{\pm s}_4\swei{\Xi}{\pm s}_{{\rm norm},4}\,, \label{eq:angular-ODE-decomposition}
\end{align}
where $\swei{\Xi}{\pm s}_{{\rm norm},i}$, for $i=1,\dots,4$, encode the two linearly independent behaviors that solutions may take at the regular singular points $\theta=0,\pi$ \cite[Chapter 5]{Olver1973}:
\begin{itemize}[noitemsep]
\item if $|m|\neq s$, take  $\swei{\Xi}{\pm s}_{{\rm norm},1}(1-\cos\theta)^{-\frac{m\pm s}{2}}$ and $\swei{\Xi}{\pm s}_{{\rm norm},2}(1-\cos\theta)^{\frac{m\pm s}{2}}$ to be smooth as $\theta\to 0$ and normalized at the $\theta=0$ end;
\item if $|m|\neq s$, take $\swei{\Xi}{\pm s}_{{\rm norm},3}(1+\cos\theta)^{-\frac{m\mp s}{2}}$ and $\swei{\Xi}{\pm s}_{{\rm norm},4}(1-\cos\theta)^{\frac{m\mp s}{2}}$ to be smooth as $\theta\to \pi$ and normalized at the $\theta=\pi$ end;
\item if $|m|=s$, take $\swei{\Xi}{+s}_{{\rm norm},4}$, $\swei{\Xi}{+s}_{{\rm norm},1}$, $\swei{\Xi}{-s}_{{\rm norm},3}$ and $\swei{\Xi}{-s}_{{\rm norm},2}$ exactly as above; the definition of the remaining functions requires a logarithmic correction which we need not specify here.
\end{itemize}

We say that $\Xi_{m\Lambda}^{[\pm s],\,(\nu)}$ is a smooth $(\pm s)$-spin-weighted function if $\swei{a}{\pm s}_4=\swei{a}{\pm s}_2=0$ when $m>s$, if $\swei{a}{\pm s}_3=\swei{a}{\pm s}_1=0$ when $m<-s$ and if $\swei{a}{+s}_3=\swei{a}{+s}_2=0=\swei{a}{-s}_4=\swei{a}{-s}_1$ when $|m|\leq s$, see also \cite[Section 2.2.1]{Dafermos2017}. The space of such solutions is somewhat small, as the following proposition indicates:

\begin{proposition}[Smooth spin-weighted spheroidal harmonics] \label{prop:angular-ode}
Fix $s\in\frac12\mathbb{Z}_{\geq 0}$, let $m-s\in\mathbb{Z}$, and assume $\nu\in\mathbb{R}$. Consider the angular ODE \eqref{eq:angular-ode} with the boundary condition that $e^{im\phi}S_{m,\bm\uplambda}^{[\pm s],\,(\nu)}$ is a non-trivial, normalized, smooth $(\pm s)$-spin-weighted function. There are countably many such solutions to this eigenvalue problem. Using $l$ as an index, we write such solutions, also called $(\pm s)$-spin-weighted spheroidal harmonics with spheroidal parameter $\nu$, as $e^{im\phi}S^{[\pm s],\, (\nu)}_{ml}$ and denote the corresponding eigenvalues, which are real and independent of the sign chosen for $s$, by $\bm\Lambda^{ (\nu)}_{sml}$. The parameter $l$ is chosen so that $l-s\in\mathbb{Z}$, $l\geq \max\{|m|,s\}$ and $\bm\Lambda^{(0)}_{sml}=l(l+1)-s^2\geq s$. The following alternative notation will also be used:
\begin{align}
\bmLslash^{(\nu)}_{ml}:= \bm\Lambda^{(\nu)}_{sml} -2m\nu+s\,. \label{eq:def-Lambda}
\end{align}
\end{proposition}

\begin{proof}
Since $\nu\in\mathbb{R}$, the result follows from standard Sturm--Liouville theory, see \cite{Meixner1954} for further details.
\end{proof}

\begin{remark} In what follows, we often lighten the notation, replacing the spin-weighted spheroidal eigenvalues $\bm\Lambda^{(\nu)}_{sml}$ by $\bm\Lambda$ and similarly for $\bmLslash$. However, non-bold characters $\Lambda$ and $\Lslash\equiv \Lambda-2m\nu+s$ are not the same: they denote real parameters which are not constrained to be the eigenvalues identified in Proposition~\ref{prop:high-freq-expansion} for some $(m,l)$.
\end{remark}

Given that \eqref{eq:angular-ode} is analytic in the coefficient $\nu$, the spin-weighted spheroidal eigenvalue ${\bmLslash}^{(\nu)}_{sml}$ admits an expansion in powers of $\nu$ as $|\nu|\to \infty$. Such expansions for $s\neq 0$ go back to \cite{Breuer1977}, but they were completed and corrected by the work of the first author and collaborators: 

\begin{proposition}[\cite{Casals2005,Casals2018}] \label{prop:high-freq-expansion} Fix $s\in\frac12\mathbb{Z}_{\geq 0}$, a number $m$ such that $m-s\in\mathbb{Z}$, a number $l$ such that $l-\max\{|m|,s\}\in\mathbb{Z}_{\geq 0}$, and $\nu\in\mathbb{R}$. Let ${\bmLslash}^{(\nu)}_{sml}$ be as in Proposition~\ref{prop:angular-ode}. Then, as $\nu\to \infty$, for any $N>0$, we can find real  constants $A_k=A_k(s,l,m,\nu)$, $k\leq N$, such that
\begin{gather}
\bmLslash_{sml}^{(\nu)}= \sum_{k=-1}^N\frac{A_k}{\nu^k}+O\lp(\nu^{-N-1}\rp)\,. \label{eq:lambda-high-freq-expansion}
\end{gather}
The first two coefficients are
\begin{align*}
A_{-1}=2(q_{sml}-m)\,, \qquad A_0= -\frac{1}{2}\lp[(q_{sml})^2-m^2+1-2s\rp]\,,
\end{align*}
where, denoting by $\mr{odd}(\cdot)$ a function on $\mb{Z}$ which is one if the argument is odd and zero otherwise,  we may express $q_{sml}$ as
\begin{align} 
q_{sml}= 
\begin{dcases}
l+1-\mr{odd}(l+m),		    \quad & \text{if} \quad l \geq |m+s |+s \\
	2l + 1 -(|m+s |+s), \quad &  \text{if} \quad l < |m+s |+s
	\end{dcases}\,. \label{eq:qmls-def}
\end{align}
The following coefficients may be computed as follows. Let
\begin{align*}
Q_{sml,n}^+ &\equiv(2n+q_{sml}+s-|m-s|+1)(2n+q_{sml}-s+|m+s|+1)\,,\\
Q_{sml,n}^- &\equiv(2n+q_{sml}+s+|m-s|-1)(2n+q_{sml}-s-|m+s|-1)\,.
\end{align*}
Then, for $k\geq 1$, we have
\begin{align}
A_k&=\frac14Q_{sml,0}^+ a_{1,k}+\frac14Q_{sml,0}^- a_{-1,k}\,, \label{eq:high-frequency-recursive-relations-Ak}
\end{align}
where $a_{1,k}$ and $a_{-1,k}$ are computed from the following recursive relations for $n\neq 0$:
\begin{align}
\begin{dcases}
a_{n,|n|}&=-\frac{1}{16n}
\lp\{\begin{array}{lr}
Q_{sml,n}^{-}a_{n-1,|n|-1}\,, &\quad n\geq 1\\
Q_{sml,n}^{+}a_{n+1,|n|-1}\,, &\quad n\leq -1
\end{array}\rp\} \,, 
\\
a_{n,|n|+1}\!\!\!\!\!\!&=\frac12 (q_{sml}+n)a_{n,|n|}
-\frac{1}{16n}
\lp\{\begin{array}{lr}
Q_{sml, n}^- a_{n-1,|n|}\,, & n\geq 1\\
Q_{sml, n}^{+}a_{n+1,|n|}\,, & n\leq -1
\end{array}\rp\} \,,
\\
a_{n,j+1}&=\frac12 (q_{sml}+n)a_{n,j}
-\frac{Q_{sml, n}^+}{16n} a_{n+1,j}
-\frac{Q_{sml, n}^-}{16n}a_{n-1,j}
+\sum_{i\geq 1}\frac{A_i}{4n} a_{n,i-j}\,, \quad j\geq |n|+1\,,
\end{dcases}
 \label{eq:high-frequency-recursive-relations-ank}
\end{align}
initialized by the choice $a_{0,0}=1$ and $a_{0,j}=0$ if $j\neq 0$. In particular, for $k\leq 7$, $A_k$ are explicitly determined in \cite[Equations 3.15--3.21]{Casals2018}. 
\end{proposition}

\begin{remark} Due to the symmetries of the angular ODE~\eqref{eq:angular-ode},  the high frequency limit $\nu\to -\infty$ follows from Proposition~\ref{prop:high-freq-expansion} by replacing $m\to -m$.
\end{remark}

\subsection{The angular Teukolsky--Starobinsky identities and constants}

In this section, we will require the notation 
\begin{align}
\hat{\mc{L}}^{\pm}_n &\equiv \frac{d}{d\theta}\pm\lp(\frac{m}{\sin\theta}-\nu\sin\theta\rp)+n\cot\theta\,,
\end{align}
where $\theta\in(0,\pi)$, $n,m\in\frac12\mathbb{Z}$, $\nu\in\mathbb{R}$.

\begin{proposition}[Angular TS identities] \label{prop:TS-angular}
Fix $s\in\frac12\mathbb{Z}_{\geq 0}$, $m$ such that $m-s\in\mathbb{Z}\backslash\{0\}$, $\Lambda\in\mathbb{R}$ and $\nu\in\mathbb{R}$. Then, if ${\Xi}^{[\pm s],\,(\nu)}_{m\Lambda}$  solve \eqref{eq:angular-ode} and admit the decomposition \eqref{eq:angular-ODE-decomposition},
\begin{align}
\begin{split}
\lp(\prod_{k=0}^{2s-1}\hat{\mc{L}}_{s-k}^+\rp)\Xi_{m\Lambda}^{[+s],\, (\nu)}&= \swei{a}{+s}_1 \mathfrak{B}^{(1)}_s \swei{\Xi}{-s}_{{\rm norm},1}+\swei{a}{+s}_2 \mathfrak{B}^{(7)}_s \swei{\Xi}{-s}_{{\rm norm},2}\\
&=\swei{a}{+s}_3 \mathfrak{B}^{(4)}_s \swei{\Xi}{-s}_{{\rm norm},3}+\swei{a}{+s}_4 \mathfrak{B}^{(6)}_s \swei{\Xi}{-s}_{{\rm norm},4}\,,
\\
\lp(\prod_{k=0}^{2s-1}\hat{\mc{L}}_{s-k}^-\rp)\Xi_{m\Lambda}^{[-s],\, (\nu)}&= \swei{a}{-s}_1 \mathfrak{B}^{(3)}_s \swei{\Xi}{+s}_{{\rm norm},1}+\swei{a}{-s}_2 \mathfrak{B}^{(5)}_s \swei{\Xi}{+s}_{{\rm norm},2}\\
&=\swei{a}{-s}_3 \mathfrak{B}^{(2)}_s \swei{\Xi}{+s}_{{\rm norm},3}+\swei{a}{-s}_4 \mathfrak{B}^{(8)}_s \swei{\Xi}{+s}_{{\rm norm},4}\,,
\end{split} \label{eq:TS-angular-general}
\end{align}
where the products on the left hand side are replaced by the identity if $s=0$ and, if $s\neq 0$, have indices increasing from left to right. Here, $\mathfrak{B}_s^{(i)}=\mathfrak{B}_s^{(i)}(\nu,m,\Lambda)$ for $i=1,\dots, 8$. Indeed, if $s=0$, $\mathfrak{B}_s^{(i)}=1$ for $i=1,\dots,8$. For $s\neq 0$, we easily obtain
\begin{align*}
(-1)^{2s}\mathfrak{B}_s^{(2)}=\mathfrak{B}_s^{(6)}=\mathfrak{B}_s^{(1)}=(-1)^{2s}\mathfrak{B}_s^{(5)}= 2^s\prod_{j=0}^{2s-1}\lp(m+s-\frac32 j\rp)\,;
\end{align*}
the other $\mathfrak{B}_s^{(i)}$ can be computed explicitly in terms of the first $s$ coefficients of the asymptotic expansions of $\swei{\Xi}{+s}_{{\rm norm},1}$, $\swei{\Xi}{+s}_{{\rm norm},3}$, $\swei{\Xi}{-s}_{{\rm norm},2}$ and $\swei{\Xi}{-s}_{{\rm norm},4}$, which in turn can be explicitly computed in terms of $(\nu,m,\Lambda)$. 
\end{proposition}

\begin{proof} The result follows from differentiating the asymptotic formulas for $\swei{\Xi}{\pm s}_{{\rm norm},i}$, $i=1,\dots, 4$. In doing so, it can be useful to note that
\begin{align*}
\prod_{k=0}^{2s-1}\hat{\mc{L}}_{s-k}^\pm=\lp(\sin\theta\rp)^{2s}\lp(\frac{\hat{\mc{L}}_{s}^\pm}{\sin\theta}\rp)^{2s} \,.
\end{align*}
For further details, we refer the reader to an analogous proof, in the setting of the radial Teukolsky--Starobinsky identities, in \cite[Proposition 2.14]{TeixeiradaCosta2019}. 
\end{proof}

We are ready to define the angular Teukolsky--Starobinsky constants:

\begin{definition}[Angular TS constants]\label{def:TS-angular-constant} Fix $s\in\frac12\mathbb{Z}_{\geq 0}$, $m$ such that $m-s\in\mathbb{Z}$, $\Lambda\in\mathbb{R}$, and $\nu\in\mathbb{R}$. 
Consider the operator 
\begin{align*}
\lp(\prod_{j=0}^{2s-1}\hat{\mc{L}}_{s-j}^\mp\rp)\lp(\prod_{k=0}^{2s-1}\hat{\mc{L}}_{s-k}^\pm\rp) = \lp(\sin\theta\rp)^{2s}\lp(\frac{\hat{\mc{L}}_{s}^\mp}{\sin\theta}\rp)^{2s}\lp(\sin\theta\rp)^{2s}\lp(\frac{\hat{\mc{L}}_{s}^\pm}{\sin\theta}\rp)^{2s}\,, 
\end{align*}
with indices $j,k$ increasing from right to left on the product, and the latter being replaced by the identity if $s=0$. If solutions of the angular ODE~\eqref{eq:angular-ode} with spin $\pm s$  are eigenfunctions of the above operator corresponding to the same eigenvalue, the eigenvalue is denoted by $\mathfrak B_s=\mathfrak B_s(a\omega,m,\Lambda)$ and it is called the {\normalfont angular Teukolsky--Starobinsky constant}. 
\end{definition}

\begin{remark} Note that, in Definition~\ref{def:TS-angular-constant}, we do not constrain $\Lambda$ to be a spin-weighted spheroidal eigenvalue, as defined in Proposition~\ref{prop:angular-ode}. In what follows, if we do take $\Lambda=\bm\Lambda_{sml}^{(\nu)}$ and $\Lslash=\bmLslash_{sml}^{(\nu)}$ for some $l$, then we write $\mathfrak{B}_s=\mathfrak B_s(\nu,m,l)$.
\end{remark}

\subsection{Examples of angular Teukolsky--Starobinsky constants}

By direct computation, we can check that an angular Teukolsky--Starobinsky  constant exists at least for low values of Teukolsky spin:
\begin{lemma} \label{lemma:TS-angular-constant-examples} For any $s\in\{0,\frac12,1,\frac32,2,\frac52,3\}$, there exists an angular Teukolsky--Starobinsky constant. Moreover, it is given by:
\begin{align}\label{eq:TS-angular-constants}
\begin{split} 
\mathfrak B_{0}(\nu, m, \Lambda)&= 1\,,\\
-\mathfrak B_{\frac12}(\nu, m, \Lambda)&= \Lslash\,,\\
\mathfrak B_{1}(\nu, m, \Lambda)&= \Lslash^2+4m\nu-4\nu^2\,,\\
-\mathfrak B_{\frac32}(\nu, m, \Lambda)&=  \Lslash^2\lp(\Lslash+1\rp)-16\nu\Lslash(\nu-m)+16\nu^2\,,\\
\mathfrak B_2(\nu, m, \Lambda)&= \Lslash^2(\Lslash+2)^2+40\nu\Lslash^2(m-\nu)+48\nu\Lslash(m+\nu)+144\nu^2(m-\nu)^2\,,\\
-\mathfrak B_{\frac52}(\nu, m, \Lambda)&= \Lslash^2\lp(\Lslash+3\rp)^2\lp(\Lslash+4\rp)+
16 m\nu \Lslash (3+\Lslash) (8+5 \Lslash) m  -16 \nu^2 \left(\Lslash
   (-12+\Lslash (2+5 \Lslash))\right) \\
   &\qquad   +1024 (1+\Lslash) \nu^2m^2   +1024 \nu^3
   (1-2 \Lslash) m +1024 \nu^4 (-2+\Lslash),  \\
\mathfrak B_{3}(\nu, m, \Lambda)&= \Lslash^2(\Lslash+4)^2(\Lslash+6)^2
+4 m\nu \Lslash (4+\Lslash) (360+7 \Lslash (36+5 \Lslash))  \\
&\qquad +4 \nu^2
   \left(-\Lslash (4+\Lslash) (-120+7 \Lslash (4+5 \Lslash))+4 (900+\Lslash
   (1140+259 \Lslash)) m^2\right) \\
   &\qquad+32 \nu^3 m \left(300-260 \Lslash-259
   \Lslash^2+450 m^2\right) +16 \nu^4 \left(100+\Lslash (-620+259
   \Lslash)\right) \\
   &\qquad -43200 \nu^4 m (m-\nu)-14400 \nu^6. 
\end{split}
\end{align}
In all the above examples, if $\nu=0$ and $\Lslash=\bmLslash_{sml}^{(\nu)}$ corresponds to a spin-weighted spheroidal eigenvalue with spheroidal parameter $\nu=0$ for some $l$, then $(-1)^{2s}\mathfrak{B}_s(\nu=0,m,l)\geq 1$.
\end{lemma} 
\begin{proof} The proof of existence is similar to the case of the radial constants, in Lemma~\ref{lemma:TS-radial-constant-examples} to come, so we do not present it here. 
For the final statement, it is easy to see that, if $\nu=0$, only the first term of each expression in \eqref{eq:TS-angular-constants} remains. The least value is attained by $-\mathfrak{B}_{1/2}(\nu=0,m,l=s=1/2) = 1$. 
\end{proof}

\begin{remark} We expect that the angular Teukolsky--Starobinsky constant can be defined for all $s\in\frac12\mathbb{Z}_{\geq 0}$. However, no proof has been given in the literature.
\end{remark}

\subsection{Properties of the angular Teukolsky--Starobinsky constants}

If they exist, $\mathfrak{B}_s$ are non-negative (integer $s$) or non-positive (half-integer $s$):

\begin{lemma}[Sign of angular TS constants] \label{lemma:TS-angular-constant-sign} If an angular Teukolsky--Starobinsky constant exists for a certain $(\nu,m,\Lambda)$ satisfying the constraints in Definition~\ref{def:TS-angular-constant}, then $(-1)^{2s}\mathfrak B_s(\nu,m,\Lambda)\geq 0$. 
\end{lemma}
\begin{proof} To show that one has $(-1)^{s}\mathfrak{B}_s\geq 0$, recall the integration by parts identity of \cite[Section 68, Lemma 4]{Chandrasekhar} (earlier in \cite{Teukolsky1974}):
for $f$ and $h$ sufficiently regular functions of $\theta$,
\begin{align} \label{eq:TS-angular-IBP-lemma}
\int_0^\pi h \lp(\mc{L}^{\pm}_n f\rp) \sin\theta d\theta = -\int_0^\pi f \lp(\mc{L}^{\mp}_{-n+1} h\rp) \sin\theta d\theta\,.
\end{align}
Without loss of generality, let $\Xi^{[\pm s]}$ be  real solutions to \eqref{eq:angular-ode} normalized to have unit $L^2$ norm. Then, by \eqref{eq:TS-angular-IBP-lemma}, assuming existence of the constant,
\begin{align}
\mathfrak{B}_s &= \int_0^\pi \Xi^{[\pm s]}\prod_{j=0}^{2s-1}\mc{L}^\mp_{s-j}\prod_{k=0}^{2s-1}\mc{L}^\pm_{s-k}\Xi^{[\pm s]} \sin\theta d\theta  = (-1)^{2s} \int_0^\pi \lp(\prod_{k=0}^{2s-1}\mc{L}^\pm_{s-k}\Xi^{[\pm s]}\rp)^2 \sin\theta d\theta\,, \label{eq:TS-angular-constant-variational-formula}
\end{align}
where the integral on the right hand side is non-negative.
\end{proof}

\begin{remark} We note that the non-negativity of $(-1)^{2s}\mathfrak{B}_s$ provides nontrivial constraints on the values that $\Lambda$ can take in terms of $m$ and $\nu$; see also Lemma~\ref{lemma:TS-radial-constant-examples} for concrete examples.
\end{remark}

We may go further when we let $\Lambda=\bm\Lambda_{sml}^{(\nu)}$ be a spin-weighted spheroidal eigenvalue for some $l$:

\begin{lemma}[Zeros of angular TS constants] \label{lemma:TS-angular-constant-zeros} Fix $s\in\frac12\mathbb{Z}_{\geq 0}$, a number $m$ such that $m-s\in\mathbb{Z}$, a number $l$ such that $l-\max\{|m|,s\}\in\mathbb{Z}_{\geq 0}$, and $\nu\in\mathbb{R}$. If an angular Teukolsky--Starobinsky constant exists for such $(\nu,m,l)$, then $(-1)^{2s}\mathfrak B_s(\nu,m,l)> 0$. 
\end{lemma}

\begin{proof}
If the angular Teukolsky--Starobinsky constant exists, then in Proposition~\ref{prop:TS-angular}, we must have
\begin{gather*}
\mathfrak{B}_s^{(3)}=\frac{\mathfrak{B}_s}{\mathfrak{B}_s^{(1)}}\,, \quad
\mathfrak{B}_s^{(4)}=\frac{\mathfrak{B}_s}{\mathfrak{B}_s^{(2)}}\,, \qquad
\mathfrak{B}_s^{(7)}= \frac{\mathfrak{B}_s}{\mathfrak{B}_s^{(5)}}\,,\quad
 \mathfrak{B}_s^{(8)}=\frac{\mathfrak{B}_s}{\mathfrak{B}_s^{(6)}}\,.
\end{gather*}

We deal first with the case  $|m|>s$, where we follow the strategy of the proof of \cite[Lemma 2.19]{TeixeiradaCosta2019}, an analogous result for the radial ODE. Suppose $m<-s$ and $\mathfrak B_s=0$; for a spin-weighted spheroidal harmonics, we have $\swei{a}{\pm s}_3=\swei{a}{\pm s}_1=0$, but by \eqref{eq:TS-angular-general},  $\swei{a}{-s}_2,\swei{a}{+s}_4=0$ too, which would force $S_{ml}^{[\pm s],\,(\nu)}\equiv 0$. Now suppose $m>s$ and $\mathfrak B_s=0$; for a spin-weighted spheroidal harmonics, we have $\swei{a}{\pm s}_4=\swei{a}{\pm s}_2=0$, but by \eqref{eq:TS-angular-general}, this implies $\swei{a}{+s}_1,\swei{a}{-s}_3=0$ too, which leads to another contradiction. 

Finally, suppose $-s\leq m\leq s$, so that a spin-weighted spheroidal harmonic has $\swei{a}{+s}_3=\swei{a}{+s}_2=0=\swei{a}{-s}_4=\swei{a}{-s}_1$. On the other hand, from the formula \eqref{eq:TS-angular-constant-variational-formula}, we find that $\mathfrak{B}_s=0$ if and only if \eqref{eq:TS-angular-general} vanish, in which case one must also have $\swei{a}{+s}_1=\swei{a}{+s}_4=0$, given that $\mathfrak{B}_s^{(1)},\mathfrak{B}_s^{(6)}\neq 0$, and that $\swei{a}{-s}_2=\swei{a}{-s}_3=0$, given that $\mathfrak{B}_s^{(3)},\mathfrak{B}_s^{(2)}\neq 0$. We conclude $S_{ml}^{[\pm s],\,(\nu)}\equiv 0$, which is a contradiction.

Hence, $\mathfrak{B}_s(\nu,m,l)\neq 0$. By Lemma~\ref{lemma:TS-angular-constant-sign}, the conclusion follows.
\end{proof}

In spite of the angular Teukolsky--Starobinsky constants' positivity, they approach arbitrarily small values, at least in the cases considered in Lemma \ref{lemma:TS-angular-constant-examples}.

\begin{lemma} \label{lemma:TS-angular-high-freq-expansion}
Fix $s\in\{\frac12,1,\frac32,2,\frac53,3\}$. Then, there are some pairs $(l,m)$, where $m-s\in\mathbb{Z}$ and $l-\max\{|m|,s\}\in\mathbb{Z}_{\geq 0}$, for which we have, as $\nu\to \infty$,
\begin{equation}
(-1)^{2s}\mathfrak{B}_s(\nu,m,l) = O(\nu^{-N})\,, \quad \forall\, N>0\,. \label{eq:TS-angular-high-freq-expansion}
\end{equation}
\end{lemma}

\begin{proof}
For $s=1/2$, the leading order term vanishes if and only if 
\begin{align*}
q_{\frac12, ml}- m=0\Leftrightarrow l=m\,,\,\, m\geq \frac12\,.
\end{align*}
For $s=1$, $\mathfrak{B}_1 = (A_{-1}^2-4)\nu^2+ O(\nu)$, where the leading order term vanishes if and only if 
\begin{align*}
q_{1,ml}- m=\pm 1\Leftrightarrow l=\lp\{\begin{array}{ll}
m, &m\geq 1\\
m+1, &m\geq 0
\end{array}\rp.\,.
\end{align*}
For $s=3/2$, the leading order term of $-\mathfrak{B}_{3/2} = (A_{-1}^2-16)A_{-1}\nu^3+ O(\nu^{2})$ vanishes if and only if 
\begin{gather*}
q_{\frac32,ml}- m=0,\pm 2
\Leftrightarrow 
l=\lp\{\begin{array}{llll}
m, & m\geq 3/2; & &\\
m+1, &m\geq 1/2; & ~~m+2, &m\geq -1/2\\
\end{array}\rp.\,.
\end{gather*}
For $s=2$, $\mathfrak{B}_2 = (A_{-1}^2-4)(A_{-1}^2-36)\nu^4+ O(\nu^{3})$, and the leading order term vanishes if and only if 
\begin{gather*}
q_{2,ml}- m=\pm 1, \pm 3
\Leftrightarrow 
l=\lp\{\begin{array}{llll}
m, &m\geq 2; &~~ m+2, &m\geq 0\\
m+1, &m\geq 1; &~~m+3, &m\geq -1
\end{array}\rp.\,.
\end{gather*}
For $s=5/2$, the leading order term of $\mathfrak{B}_{\frac52}$ vanishes if and only if 
\begin{gather*}
q_{\frac52,ml}- m=0,\pm 2, \pm 4
\Leftrightarrow l=\lp\{\begin{array}{llll}
m, & m\geq 5/2; & &\\
m+1, &m\geq 3/2; & ~~m+3, &m\geq -1/2\\
m+2, &m\geq 1/2; & ~~m+4, &m\geq -3/2
\end{array}\rp.\,.
\end{gather*}
For $s=3$, the leading order term  of $\mathfrak{B}_3$ vanishes if and only if 
\begin{gather*}
q_{3,ml}- m=\pm 1, \pm 3, \pm 5
\Leftrightarrow l=\lp\{\begin{array}{llll}
m, & m\geq 3; & ~~m+3 &m\geq 0\\
m+1, &m\geq 2; & ~~m+4, &m\geq -1\\
m+2, &m\geq 1; & ~~m+5, &m\geq -2
\end{array}\rp.\,.
\end{gather*}

To obtain \eqref{eq:TS-angular-high-freq-expansion} for some $N>0$, we need to compute the coefficients $A_k$ of the asymptotic expansion \eqref{eq:lambda-high-freq-expansion} up to $k=N-1+2s$. The formulas in \cite[Equation 3.15--3.21]{Casals2018} for $A_k$ up to $k=7$, for instance, yield that \eqref{eq:TS-angular-high-freq-expansion} holds for $N\geq 7+1-2s$ for all $(s,l,m)$ identified in the preceding paragraph.

In the cases
\begin{gather}
\begin{gathered}
s=\frac12\,,\,\, l=m\geq 1/2\,; \qquad 
s=1\,,\,\, l=1\,,\,\, m=0\,; \qquad 
s=\frac32\,,\,\, l=\frac32\,,\,\, m=\frac12\,;\\
s=2\,,\,\, l=2\,,\,\, m=-1\,;\qquad 
s=\frac52\,,\,\, l=\frac52\,,\,\, m=-\frac32\,; \qquad
s=3\,,\,\, l=3\,,\,\, m=-2\,;
\end{gathered}\label{eq:sml-for-angularTS-superpolynomial-decay}
\end{gather}
the formulas \cite[Equation 3.15--3.21]{Casals2018} in fact imply that $A_{-1}=\pm 2$ and that $A_k=0$ for $k=0,\dots\, 7$. By the recursive formulas \eqref{eq:high-frequency-recursive-relations-Ak} and \eqref{eq:high-frequency-recursive-relations-ank}, we must have that $A_k=0$ for all $k\geq 1$. Hence, for such $(s,m,l)$, \eqref{eq:TS-angular-high-freq-expansion} holds for all $N>0$, as stated.
\end{proof}

\begin{figure}[b]
\centering
\includegraphics[width=8.5cm]{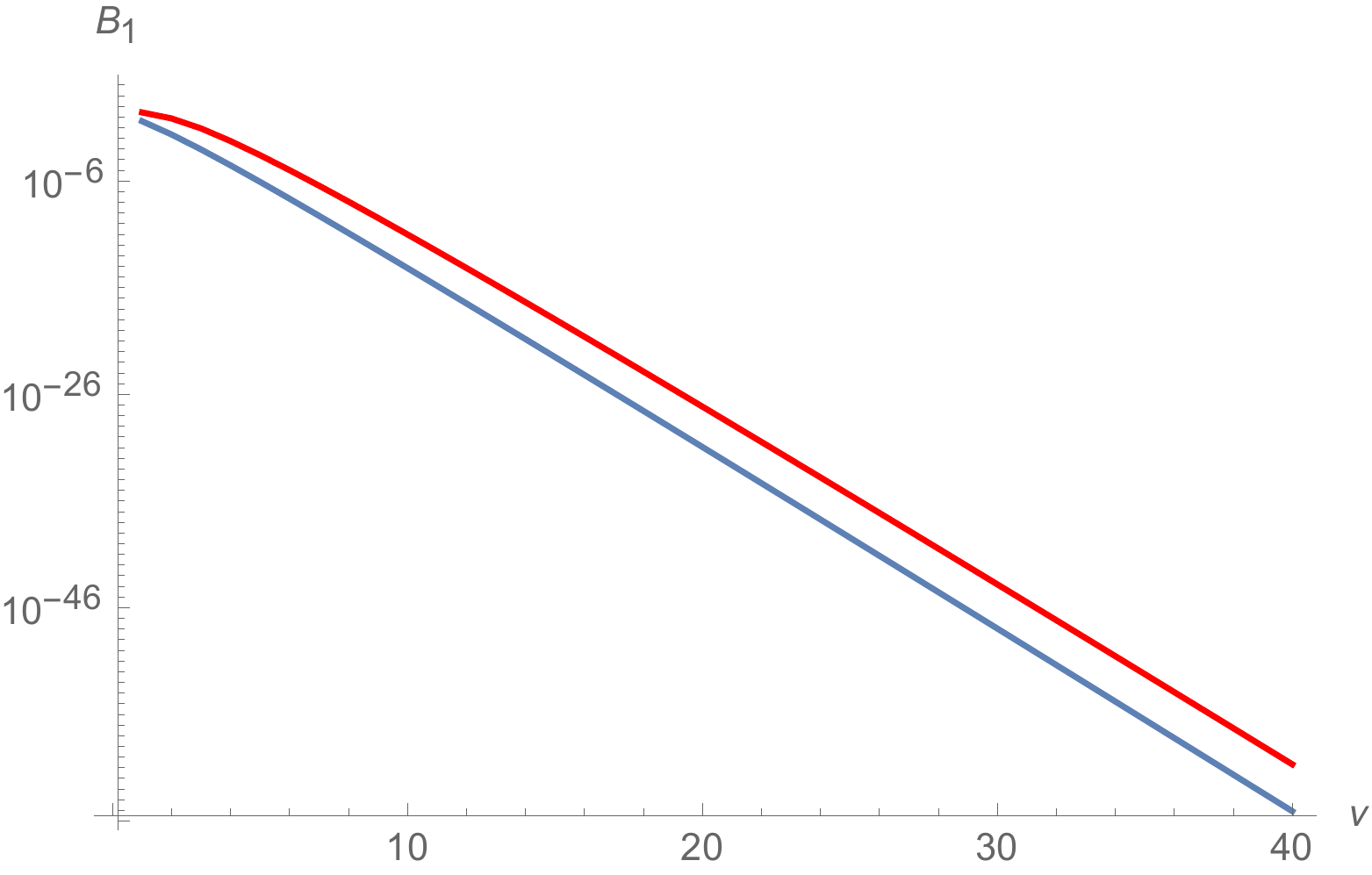}
\caption{Logarithmic plot of $\mathfrak{B}_{1}(\nu,m, l=1)$ as a function of $\nu$ for $m=1$ (blue) and $m=0$ (red), obtained by numerically computing $\Lambda_{sml}^{(\nu)}$ via \cite{BHPToolkit}. The linear relation in the logarithmic plot suggests exponential decay of $\mathfrak{B}_{1}$ as $\nu\to \infty$ for these modes.}
\label{fig:B1}
\end{figure}

\begin{remark} Computer-based symbolic computations suggest that, in fact, all of the triples $(s,m,l)$ identified in the first paragraph of the proof of Lemma~\ref{lemma:TS-angular-high-freq-expansion} verify the conclusion of the lemma, rather than just the smaller set in \eqref{eq:sml-for-angularTS-superpolynomial-decay}. This is confirmed by numerical analysis, see Figure~\ref{fig:B1}. However, to establish such a result would require a much more detailed analysis of the recursive relations \eqref{eq:high-frequency-recursive-relations-Ak} and \eqref{eq:high-frequency-recursive-relations-ank} than we pursue here.
\end{remark}

In the case $s=1$, numerical computations suggest that the superpolynomial decay with $\nu$ identified in Lemma~\ref{lemma:TS-angular-high-freq-expansion} for some $(l,m)$ is actually  exponential, see Figure~\ref{fig:B1}.

\section{The radial Teukolsky--Starobinsky constants}
\label{sec:radial-TS-constants}

In this section, we introduce the radial ODE corresponding to the Teukolsky equation~\eqref{eq:angular-ode} and define the radial Teukolsky--Starobinsky constants. The last subsection contains the proof of Facts~\ref{fact:negativity} to \ref{fact:lower-bounds}.

\subsection{The radial ODE}

We consider the radial ODEs
\begin{align*}
&\lp[\Delta^{\mp s}\frac{d}{dr}\lp(\Delta^{\pm s+1}\frac{d}{dr}\rp)  +\frac{[\omega(r^2+a^2)-am]^2-2i(\pm s)(r-M)[\omega(r^2+a^2)-am]}{\Delta}\rp] \smlambda{\upalpha}{\pm s}(r)\\
&\qquad\quad+\lp(\pm 4is\omega r -\Lambda\mp s+2am\omega\rp)\smlambda{\upalpha}{\pm s}(r)=0\,,\numberthis \label{eq:radial-ODE-alpha}
\end{align*}
where $M>0$, $|a|\leq M$, $r\in(r_+,\infty)$, $s\in\frac12\mathbb{Z}_{\geq 0}$, $m$ such that $m-s\in\mathbb{Z}$, $\Lambda\in\mathbb{R}$,  and $\omega\in\mathbb{R}$. An asymptotic analysis of \eqref{eq:radial-ODE-alpha} shows that a solution to \eqref{eq:radial-ODE-alpha} admits a unique set of complex numbers $\swei{a}{\pm s}_{\mc{H}^+}$, $\swei{a}{\pm s}_{\mc{H}^-}$, $\swei{a}{\pm s}_{\mc{I}^+}$ and $\swei{a}{\pm s}_{\mc{I}^-}$ such that
\begin{align}
\smlambda{\upalpha}{\pm s}=\swei{a}{\pm s}_{\mc{H}^+}\cdot\swei{\upalpha}{\pm s}_{\mc{H}^+}+\swei{a}{\pm s}_{\mc{H}^-}\cdot\swei{\upalpha}{\pm s}_{\mc{H}^-} = \swei{a}{\pm s}_{\mc{I}^+}\cdot\swei{\upalpha}{\pm s}_{\mc{I}^+}+\swei{a}{\pm s}_{\mc{I}^-}\cdot\swei{\upalpha}{\pm s}_{\mc{I}^-}\,, \label{eq:radial-ODE-decomposition}
\end{align}
where $\swei{\upalpha}{\pm s}_{\mc{I}^+}$ and $\swei{\upalpha}{\pm s}_{\mc{I}^-}$ are solutions of \eqref{eq:radial-ODE-alpha} with outgoing and ingoing, respectively, boundary conditions as $r\to \infty$ which are normalized at $r=\infty$, and where $\swei{\upalpha}{\pm s}_{\mc{H}^+}$ and $\swei{\upalpha}{\pm s}_{\mc{H}^-}$ are solutions of \eqref{eq:radial-ODE-alpha} with ingoing and outgoing, respectively, boundary conditions as $r\to r_+$ which are normalized at $r=r_+$. Concretely,
\begin{itemize}
\item $\swei{\upalpha}{\pm s}_{\mc{I}^+}e^{-i\omega r}r^{-2iM\omega+1\pm 2s}$ and $\swei{\upalpha}{\pm s}_{\mc{I}^-}e^{i\omega r}r^{2iM\omega+1}$ are smooth functions of $1/r$ as $r\to \infty$ which are normalized at $r=\infty$;
\item if $|a|<M$, $\swei{\upalpha}{\pm s}_{\mc{H}^+}(r-r_+)^{i\frac{2Mr_+}{r_+-r_-}(\omega-m\upomega_+)\pm s}$ and $\swei{\upalpha}{\pm s}_{\mc{H}^-}(r-r_+)^{-i\frac{2Mr_+}{r_+-r_-}(\omega-m\upomega_+)}$ are smooth as $r\to r_+$ and normalized at $r=r_+$;
\item if $|a|=M$, $\swei{\upalpha}{\pm s}_{\mc{H}^+}(r-M)^{2iM\omega \pm 2s}e^{-i\frac{2M^2}{r-M}(\omega-m\upomega_+)}$ and $\swei{\upalpha}{\pm s}_{\mc{H}^-}(r-M)^{-2iM\omega}e^{i\frac{2M^2}{r-M}(\omega-m\upomega_+)}$ are smooth in $1/(r-M)$ as $r\to M$ and normalized at $r=M$.
\end{itemize}

\subsection{The radial Teukolsky--Starobinsky identities and constants}

In this section, we will require the notation 
\begin{align}
\hat{\mc{D}}^{\pm}_n &\equiv \frac{d}{dr}\pm i\lp(\frac{\omega(r^2+a^2)}{\Delta} -\frac{am}{\Delta}\rp)+\frac{2n(r-M)}{\Delta}\,,
\end{align}
where $M>0$, $|a|\leq M$, $r\in(r_+,\infty)$, $n,m\in\frac12\mathbb{Z}$ and $\omega\in\mathbb{R}$.

We quote from \cite[Proposition 2.14]{TeixeiradaCosta2019} the following:

\begin{proposition}[Radial TS identities] \label{prop:TS-radial} Fix $M>0$, $|a|\leq M$ and $s\in\frac12\mathbb{Z}_{\geq 0}$, $m$ such that $m-s\in\mathbb{Z}$, $\Lambda\in\mathbb{R}$. Then, if $\swei{\upalpha}{\pm s}$ solve \eqref{eq:radial-ODE-alpha} and admit the decomposition \eqref{eq:radial-ODE-decomposition},
\begin{align}
\begin{split}
\Delta^s \lp(\hat{\mc{D}}_0^{+}\rp)^{2s}\lp(\Delta^s \swei{\upalpha}{+ s}\rp)&=\swei{a}{+s}_{\mc{I}^{+}}\mathfrak{C}_s^{(1)}\swei{\upalpha}{- s}_{\mc{I}^+}+\swei{a}{+s}_{\mc{I}^{-}}\mathfrak{C}_s^{(7)}\swei{\upalpha}{- s}_{\mc{I}^-}=\swei{a}{+s}_{\mc{H}^{+}}\mathfrak{C}_s^{(4)}\swei{\upalpha}{- s}_{\mc{H}^+}+\swei{a}{+s}_{\mc{H}^{-}}\mathfrak{C}_s^{(6)}\swei{\upalpha}{- s}_{\mc{H}^-}\,, \\
 \lp(\hat{\mc{D}}_0^{-}\rp)^{2s}\swei{\upalpha}{- s}&=\swei{a}{-s}_{\mc{I}^{+}}\mathfrak{C}_s^{(3)}\swei{\upalpha}{+ s}_{\mc{I}^+}+\swei{a}{+s}_{\mc{I}^{-}}\mathfrak{C}_s^{(5)}\swei{\upalpha}{+ s}_{\mc{I}^-}=\swei{a}{-s}_{\mc{H}^{+}}\mathfrak{C}_s^{(2)}\swei{\upalpha}{+ s}_{\mc{H}^+}+\swei{\upalpha}{+ s}_{\mc{H}^{-}}\mathfrak{C}_s^{(8)}\swei{\upalpha}{+ s}_{\mc{H}^-}\,,
\end{split}\label{eq:TS-radial-general}
\end{align}
where the products on the left hand side are replaced by the identity if $s=0$ and, if $s\neq 0$, have indices increasing from left to right.
Here, $\mathfrak{C}_s^{(i)}=\mathfrak{C}_s^{(i)}(a,M,\omega,m,\Lambda)$ for $i=1,\dots, 8$. Indeed, if $s=0$, $\mathfrak{C}_s^{(i)}=1$ for $i=1,\dots,8$. For $s\neq 0$, we easily obtain
\begin{gather*}
\mathfrak{C}_s^{(2)}=\prod_{j=0}^{2s-1}\lp[-4iMr_+(\omega-m\upomega_+)+(s-j)(r_+-r_-)\rp] \begin{dcases}
1 &\text{~if~} |a|=M\\
(r_+-r_-)^{-2|s|}&\text{~if~} |a|<M
\end{dcases}\,, \\
\mathfrak{C}_s^{(1)}=(2i\omega)^{2s}\,, \quad\mathfrak{C}_s^{(5)}=(-2i\omega)^{2s}\,,\quad
\mathfrak{C}_s^{(6)}=
\prod_{j=0}^{2s-1}\lp[4iMr_+(\omega-m\upomega_+)+(s-j)(r_+-r_-)\rp]\,;
\end{gather*}
the remaining $\mathfrak{C}_s^{(i)}$ can be computed explicitly in terms of the first $s$ coefficients in the asymptotic expansions of $\swei{\upalpha}{\pm s}_{\mc{I}^\mp}$ and $\swei{\upalpha}{\mp s}_{\mc{H}^\mp}$, which in turn can be explicitly computed in terms of $(a,M)$ and $(\omega,m,\Lambda)$. 
\end{proposition}

We are now ready to define the radial Teukolsky--Starobinsky constants:

\begin{definition}[Radial TS constants]\label{def:TS-radial-constant} Fix $s\in\frac12\mathbb{Z}_{\geq 0}$, $m$ such that $m-s\in\mathbb{Z}$, $\Lambda\in\mathbb{R}$, $M>0$, $|a|\leq M$ and $\omega\in\mathbb{R}$. Consider the operator 
\begin{align*}
\Delta^s\lp(\hat{\mc{D}}^{\mp}_0\rp)^{2s}\lp[\Delta^s\lp(\hat{\mc{D}}^{\pm}_0\rp)^{2s}\rp] \equiv \prod_{j=0}^{2s-1}\lp(\Delta^{1/2}\hat{\mc{D}}^{\mp}_{j/2}\rp)\prod_{k=0}^{2s-1}\lp(\Delta^{1/2}\hat{\mc{D}}^{\pm}_{k/2}\rp)\,, 
\end{align*}
with indices $j,k$ increasing from right to left on the product, and the latter being replaced by the identity if $s=0$. If $\Delta^{\frac{s}{2}(1\pm 1)}\smlambda{\upalpha}{\pm s}$, where  $\smlambda{\upalpha}{\pm s}$  solve the radial ODE~\eqref{eq:radial-ODE-alpha} of spin $\pm s$, are eigenfunctions of the above operator corresponding to the same eigenvalue, the eigenvalue is denoted by $\mathfrak C_s=\mathfrak C_s(a,M,\omega,m,\Lambda)$ and it is called the {\normalfont radial Teukolsky--Starobinsky constant}. 
\end{definition}

\begin{remark} Note that, in Definition~\ref{def:TS-radial-constant} and Lemma~\ref{lemma:TS-radial-constant-examples} below, we once again do not constrain $\Lambda$ to be an eigenvalue of the angular ODE~\eqref{eq:angular-ode} with spheroidal parameter $\nu=a\omega$. In what follows, if we do take $\Lambda=\bm\Lambda_{sml}^{(a\omega)}$ and $\Lslash=\bmLslash_{sml}^{(a\omega)}$ for some $l$, then we write $\mathfrak C_s(a,M,\omega,m,l)$.
\end{remark}

\subsection{Examples of radial Teukolsky--Starobinsky constants}

By direct computation, we can check that a radial Teukolsky--Starobinsky  constant exists at least for low values of Teukolsky spin:

\begin{lemma}\label{lemma:TS-radial-constant-examples} For any $s\in\{0,\frac12,1,\frac32,2,\frac52,3\}$, there exists a radial Teukolsky--Starobinsky constant. Furthermore, it can be computed explicitly, for instance:
\begin{align}\label{eq:TS-radial-constants}
\begin{split} 
\mathfrak C_{0}(a,M,\omega,m,\Lambda)&=1\,,\\
\mathfrak C_{\frac12}(a,M,\omega,m,\Lambda)&= -\mathfrak{B}_{\frac12}(a\omega,m,\Lambda)\,,\\
\mathfrak C_{1}(a,M,\omega,m,\Lambda)&= \mathfrak{B}_{1}(a\omega,m,\Lambda)\,,\\
\mathfrak C_{\frac32}(a,M,\omega,m,\Lambda)&=  -\mathfrak{B}_{\frac32}(a\omega,m,\Lambda)\,,\\
\mathfrak C_2(a,M,\omega,m,\Lambda)&= \mathfrak{B}_{2}(a\omega,m,\Lambda)+144M^2\omega^2\,,\\
\mathfrak C_{\frac52}(a,M,\omega,m,\Lambda)&= -\mathfrak{B}_{\frac52}(a\omega,m,\Lambda) +1152(\Lslash+2)M^2\omega^2 \\
\mathfrak C_{3}(a,M,\omega,m,\Lambda)&= \mathfrak{B}_3(a\omega,m,\Lambda) + 576\lp[(3\Lslash+10)^2+100a\omega(m-a\omega)\rp]M^2\omega^2\,,
\end{split}
\end{align}
where $\mathfrak{B}_s$ may be read off from \eqref{eq:TS-angular-constants}. 
In the above, if $a\omega=0$, and $\Lslash=\bmLslash_{sml}^{(a\omega)}$ corresponds to a spin-weighted spheroidal eigenvalue with spheroidal parameter $\nu=a\omega$ for some $l$, then $\mathfrak{C}_s(a,M,\omega,m,l)\geq 1$.
\end{lemma}

\begin{proof} There are several ways of showing existence of the radial Teukolsky--Starobinsky constants. One option is to use Definition~\ref{def:TS-radial-constant}, i.e.\ to  apply the operators in the products above to radial functions one by one and using the radial ODE~\eqref{eq:radial-ODE-alpha} to trade second order derivatives of those functions by first and zeroth order terms (see for instance \cite[Sections 70 and 81]{Chandrasekhar} for $s=1, 2$). With the aid of a standard laptop, this naive approach allows one to verify existence of the constant beyond the upper bound  $s=3$ of the statement. 

Alternatively, in light of the radial Teukolsky--Starobinsky identities of Proposition~\ref{prop:TS-radial}, if the radial Teukolsky--Starobinsky constant exists, we must have
\begin{gather*}
\mathfrak{C}_s^{(3)}=\frac{\mathfrak{C}_s}{\mathfrak{C}_s^{(1)}}\,, \quad
\mathfrak{C}_s^{(4)}=\frac{\mathfrak{C}_s}{\mathfrak{C}_s^{(2)}}\,, \qquad
\mathfrak{C}_s^{(7)}= \frac{\mathfrak{C}_s}{\mathfrak{C}_s^{(5)}}\,,\quad
 \mathfrak{C}_s^{(8)}=\frac{\mathfrak{C}_s}{\mathfrak{C}_s^{(6)}}\,.
\end{gather*}
Hence, one may compute $\mathfrak{C}_s$  from one of $\mathfrak{C}_s^{(3)}, \mathfrak{C}_s^{(4)}, \mathfrak{C}_s^{(7)}, \mathfrak{C}_s^{(8)}$. As the latter are computable from the recursive formulas which yield the first $s$ coefficients of certain asymptotic series for solutions of \eqref{eq:radial-ODE-alpha}, this method is less computationally demanding than the previous one and has been suggested earlier in \cite{Fiziev2009} (see also the companion paper \cite{Fiziev2010}).

To conclude, we note that, for $\omega=0$, a glance at the formulas gives $\mathfrak{C}_s(a,M,\omega=0,m,l)=(-1)^{2s}\mathfrak{B}_s(a\omega=0,m,l)\geq 1$, from Lemma~\ref{lemma:TS-angular-constant-examples}.
\end{proof}

We remark that the formulas for $\mathfrak{C}_s$ when $s\leq 3$, given here in Lemma~\ref{lemma:TS-radial-constant-examples}, have been obtain before in \cite{Chandrasekhar1990} and \cite{Kalnins1989}.

\subsection{Properties of the radial Teukolsky--Starobinsky constants}

\subsubsection{Why the complex-conjugation argument for non-negativity is false}
\label{sec:classical-argument-nonnegativity-is-false}

This section examines the argument in \cite{Kalnins1989,Kalnins1992}, subsequently picked up by other authors, purportedly showing non-negativity of the Teukolsky--Starobinky constant for general $s\in\frac12\mathbb{Z}_{\geq 0}$. These authors claim that non-negativity of $\mathfrak{C}_s$ may be viewed as a consequence of $\Delta^{\frac{s}{2}(1\pm 1)}\smlambda{\upalpha}{\pm s}$ satisfying complex conjugate equations, or of the radial Teukolsky--Starobinsky identities being generated by complex conjugate operators, somehow implies non-negativity of the Teukolsky--Starobinky constant. To the best of our knowledge, the underlying rationale is that there is an analogy with the angular setting, where the relation between $\mc{L}_n^\pm $ and $\mc{L}_{-n+1}^\mp$ in \eqref{eq:TS-angular-IBP-lemma} leads to the conclusion of Lemma~\ref{lemma:TS-angular-constant-sign}. However, \eqref{eq:TS-angular-IBP-lemma} establishes an \textit{adjointness} relation between the angular operators $\mc{L}_n^\pm $ and $-\mc{L}_{-n+1}^\mp$ in the space of (real-valued) smooth spin-weighted functions.  In contrast, the relation between $\mc{D}_n^\pm $ and $\mc{D}_{n}^\mp$ is one of \textit{complex-conjugation}.

It is a fact of life that, in a space of complex-valued functions, as solutions to the radial ODE~\eqref{eq:radial-ODE-alpha} are bound to lie in, the complex conjugate and the adjoint of an operator are not necessarily the same\footnote{This is true already for matrices: the adjoint of a complex-valued matrix is obtained by complex-conjugation followed by transposition. Performing only one of these operations on the matrix will not, in general, produce its adjoint.}.  Indeed, for some weight $w(r)\colon(r_+,\infty)\to[0,\infty)$, for $f$ and $h$ sufficiently regular complex-valued functions of $r$ that the boundary terms in the following vanish, we have
\begin{align*}
\int_{r_+}^\infty \overline{h}\, \mc{D}^{\pm}_{0} f w \,\d r 
= -\int_{r_+}^\infty f\overline{\lp(\mc{D}^{\pm}_{0}+\frac{d}{dr}\log w\rp) h}w\, \d r\,,
\end{align*}
so the adjoint of $\mc{D}^{\pm}_{0}$ will be $-\mc{D}^{\pm}_{0}-\frac{d}{dr}\log w$. This makes the notation used in classical references such as \cite{Chandrasekhar} (see also the more recent \cite{Teukolsky2015}) rather unfortunate if one is looking to extrapolate from the contents of this book to $s>2$.

The upshot is that there is, in fact, no hope of establishing an analogue of Lemma~\ref{lemma:TS-angular-constant-sign} by a similar method. For we could establish
\begin{align*}
\mathfrak{C}_s(a,M,\omega,\Lambda) &= \int_{r_+}^\infty \overline{\upalpha^{[\pm s]}}\prod_{j=0}^{2s-1}\lp(\Delta^{1/2}\mc{D}^\mp_{j/2}\rp)\prod_{k=0}^{2s-1}\lp(\Delta^{1/2}\mc{D}^\pm_{k/2}\rp)\upalpha^{[\pm s]} w \,\d r  \\
&=  \int_{r_+}^\infty \lp|\prod_{k=0}^{2s-1}\lp(\Delta^{1/2}\mc{D}^\pm_{k/2}\rp)\upalpha^{[\pm s]}\rp|^2 w \, \d r \geq 0\,,
\end{align*}
where $\upalpha^{[\pm s]}$ is chosen to have unit $L^2_w$ norm, only if we were to prove
\begin{align}  \label{eq:TS-radial-IBP-lemma-fake}
\int_{r_+}^\infty \overline{h} \Delta^{1/2}\mc{D}^{\pm}_{n/2} f w \,\d r = \int_{r_+}^\infty f \Delta^{1/2}\overline{\mc{D}^{\mp}_{s-\frac{n+1}{2}} h} w \,\d r\,,
\end{align}
and \eqref{eq:TS-radial-IBP-lemma-fake} clearly cannot hold for any positive weight $w$. It follows that, if $\mathfrak{C}_s$ is indeed non-negative for some $s$, a different proof strategy should be sought. 

\subsubsection{Saving and improving non-negativity for $s\leq 2$}
\label{sec:saving-non-negativity}

As we have noted that the angular Teukolsky--Starobinsky constants have a definite sign (Lemma~\ref{lemma:TS-angular-constant-sign}), it is natural to try to compare the explicit expressions for such constants with those of the radial ones, i.e.\ Lemmas~\ref{lemma:TS-angular-constant-examples} and \ref{lemma:TS-radial-constant-examples} in the case where the spheroidal parameter in the angular ODE is $\nu=a\omega$. As is clear from our \eqref{eq:TS-radial-constants} (see also \cite{Kalnins1992}),
\begin{lemma} \label{lemma:TS-radial-constant-decomposition} Fix $s\in\lp\{0,\frac12,1,\frac32,2,\frac52,3\rp\}$, $m$ such that $m-s\in\mathbb{Z}$, $\Lambda\in\mathbb{R}$, $M>0$, $|a|\leq M$ and $\omega\in\mathbb{R}$. Then, there is a real $\mathfrak F_s=\mathfrak{F}_s(a\omega, m,\Lambda)$ such that
$$\mathfrak{C}_s(a,M,\omega,m,\Lambda)=(-1)^{2|s|}\mathfrak{B}_s(a\omega,m,\Lambda) + \mathfrak{F}_s(a\omega, m,\Lambda)M^2\omega^2\,.$$
Indeed, one has $\mathfrak{F}_s\equiv 0$ for $s\leq \frac32$ and 
\begin{align*}
\mathfrak{F}_2=144\,, \qquad \mathfrak{F}_{\frac52} = 1152(\Lslash+2)\,, \qquad \mathfrak{F}_{3} = 576\lp[(3\Lslash+10)^2+100a\omega(m-a\omega)\rp]\,.
\end{align*}
\end{lemma}
Lemma~\ref{lemma:TS-radial-constant-decomposition} indeed yields, in the restricted case $s\leq 2$, non-negativity of the Teukolsky--Starobinsky constant, as correctly noted in Teukolsky's original paper on the identities \cite{Teukolsky1974}. Indeed, by Lemma~\ref{lemma:TS-angular-constant-zeros}, it even yields positivity:

\begin{lemma}[Positivity of radial TS constant for $s\leq 2$] \label{lemma:TS-radial-constant-nonnegative} Fix $s\in\lp\{0,\frac12, 1,\frac32, 2\rp\}$, $m$ such that $m-s\in\mathbb{Z}$, $\Lambda\in\mathbb{R}$, $M>0$, $|a|\leq M$ and $\omega\in\mathbb{R}$. Then, $\mathfrak{F}_s(a\omega,m,\Lambda)\geq 0$ hence $\mathfrak{C}_s(a,M,\omega,m,\Lambda)\geq 0$.

Furthermore, if $\Lambda=\bm\Lambda_{sml}^{(a\omega)}$ is a spin-weighted spheroidal eigenvalue for some $l\in\mathbb{Z}_{\geq \max\{|m|,s\}}$, then $\mathfrak{C}_s(a,M,\omega,m,l)> 0$. In particular, there are no real algebraically special frequencies $(\omega,m,l)$ for any Kerr parameters $(a,M)$. 
\end{lemma}

\begin{remark} For $s> 2$,  $\mathfrak{F}_{s}$ depends nontrivially on  $\Lambda$. Without more constraints on $\Lambda$ in terms of the black hole parameters $(a,M)$ and frequencies $\omega$ and $m$, one cannot hope to investigate the validity of Lemma~\ref{lemma:TS-radial-constant-nonnegative} for $s>2$. 
\end{remark}

In fact, for $s\leq 2$, we can use Lemma~\ref{lemma:TS-radial-constant-decomposition} and the last statement in Lemma~\ref{lemma:TS-radial-constant-examples}  to obtain a more precise statement when $\Lambda=\bm\Lambda$ is a spin-weighted spheroidal eigenvalue:

\begin{lemma}[Lower bound for radial TS constant for $s=2$] \label{lemma:TS-radial-positive-spin2} Fix $s=2$, $m\in\mathbb{Z}$, $l\in\mathbb{Z}_{\geq \max\{|m|,2\}}$, $M>0$, $|a|\leq M$ and $\omega\in\mathbb{R}$. The radial Teukolsky--Starobinsky constant $\mathfrak{C}_2$ admits a positive lower bound, i.e.\ there is a $b>0$ such that 
$$\inf_{(a,M,\omega,m,l)}\mathfrak{C}_2(a,M,\omega,m,l)\geq b>0\,.$$ 
Numerically, using \cite{BHPToolkit}, we find $b\approx 150$ is enough.
\end{lemma}

\begin{proof} 
First note that, once $a$ and $(s,m,l)$ are fixed, $\bmLslash_{ml}^{[s],\,(a\omega)}$ is continuous in $\omega$ (see, for instance, \cite{Meixner1954} or \cite{Hartle1974}), hence  $\mathfrak{C}_2(a,M,\omega,m,l)$ is also continuous in $\omega$. 

At $\omega=0$, $\mathfrak{C}_2(a,M,0,m,l)=(-1)^{4}\mathfrak{B}_2(0,m,l)=(24)^2=576>0$. Hence, by continuity in $\omega$, there is a $\delta>0$ such that $\mathfrak{C}_2(a,M,\omega,m,l)\geq 1$ for $|M\omega|\leq \delta$. On the other hand, if $|M\omega|>\delta$, by the non-negativity of $\mathfrak{B}_2$ (Lemma~\ref{lemma:TS-angular-constant-sign})  $\mathfrak{C}_2\geq \mathfrak{F}_2M^2\omega^2> 144\delta^2$.  This concludes the proof.
\end{proof}

However, Lemma~\ref{lemma:TS-radial-positive-spin2} does not extend to $s< 2$. Though always positive, it follows from Lemma~\ref{lemma:TS-angular-high-freq-expansion} that the radial Teukolsky--Starobinksy constants for those spins asymptotically approach $0$ as $\omega\to \infty$ if $a\neq 0$ is fixed and suitable $(l,m)$ are chosen:
\begin{lemma} \label{lemma:TS-radial-3/2-high-freq-expansion}
Fix $M>0$, $0<|a|\leq M$ and $s\in\{\frac12,1,\frac32\}$. Then, there are some pairs $(l,m)$, where $m-s\in\mathbb{Z}$ and $l-\max\{|m|,s\}\in\mathbb{Z}_{\geq 0}$, for which we have, as $\omega\to \infty$,
\begin{equation}
\mathfrak{C}_s(a,M,\omega,m,l) = O(|M\omega|^{-N})\,, \quad \forall\, N>0\,. \label{eq:TS-radial-3/2-high-freq-expansion}
\end{equation}
\end{lemma}

\subsubsection{A myth debunked: negative values for $s\geq \frac52$}
\label{sec:myth-debunked}

Not only is the pervasive argument outlined in Section~\ref{sec:classical-argument-nonnegativity-is-false} or in \cite{Kalnins1989,Chandrasekhar1990,Kalnins1992} purportedly asserting non-negativity of $\mathfrak{C}_s$ for all $s$ incorrect, but its conclusion is also false. Lemma~\ref{lemma:TS-angular-high-freq-expansion} is our starting point to debunk this myth; we now consider the contribution of $\mathfrak{F}_s(a\omega,m,l)$:

\begin{lemma}[Negativity for $s\geq \frac52$] \label{lemma:TS-radial-constant-negative} Fix $M>0$, $0<|a|\leq M$ and $s\in\{\frac52,3\}$. Then, there are some pairs $(l,m)$, where $m-s\in\mathbb{Z}$ and $l-\max\{|m|,s\}\in\mathbb{Z}_{\geq 0}$, for which there is an $A>0$ such that, as $\omega\to \infty$,
$$\mathfrak{C}_s(a,M,\omega,m,l)=-A|M\omega|^{2|s|-2}+O(|M\omega|^{2|s|-3})\,.$$
\end{lemma}
\begin{proof} We use Proposition~\ref{prop:high-freq-expansion} once more, appealing to the proof of Lemma~\ref{lemma:TS-angular-high-freq-expansion}. Indeed, if $s=5/2$,
$\mathfrak{F}_\frac52 = 2304(q_{\frac52,ml}- m)+O(1)$, 
and we note that
\begin{align*}
q_{\frac52,ml}- m=-2 &\Rightarrow \mathfrak{C}_\frac52 = -4608\frac{a}{M}|M\omega|^3+O(|M\omega|^2)\,,\\
q_{\frac52,ml}- m=-4 &\Rightarrow \mathfrak{C}_\frac52 = -9216\frac{a}{M}|M\omega|^3+O(|M\omega|^2)\,.
\end{align*}
For $s=3$, $
\mathfrak{F}_3=576\lp[36(q_{3,ml}- m)^2-100\rp](a\omega)^2+ O(|M\omega|)$, 
and we note that
\begin{align*}
q_{3,ml}- m=\pm 1 &\Rightarrow \mathfrak{B}_3 = -36864\frac{a^2}{M^2}|M\omega|^4+O(|M\omega|^3)\,;
\end{align*}
see also Figure~\ref{fig:F3}. As shown in the proof of Lemma~\ref{lemma:TS-angular-high-freq-expansion}, the above conditions may be realized for pairs $(l,m)$ satisfying the constraints of the present lemma.
\end{proof}

\begin{figure}[htbp]
\centering
\includegraphics[width=8cm]{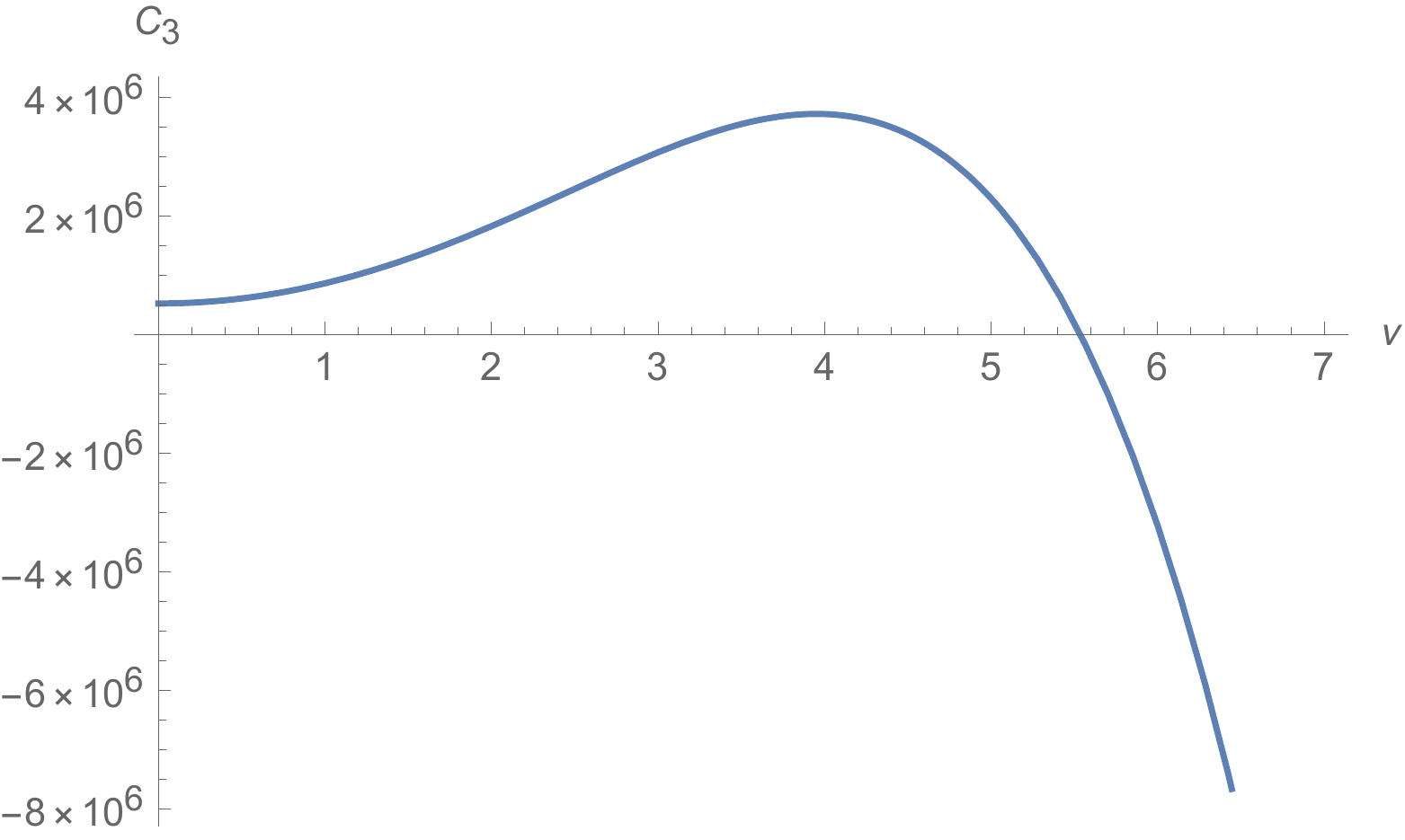}
\caption{Plot of $\mathfrak{C}_{3}(a\omega,m=0,l=3)$ as a function of $\nu=a\omega$ for $a=0.8M$, and obtained through numerical computation of the $\Lambda^{[s],(\nu)}_{ml}$.}
\label{fig:F3}
\end{figure}

While for $s=\frac52,3$ $\mathfrak{C}_s$ may take negative values for some $(\omega,m,l)$, they can also take positive values, for instance if $\omega=0$ (see Lemma~\ref{lemma:TS-radial-constant-examples}). Appealing once more to continuity in $\omega$ (see proof of Lemma~\ref{lemma:TS-radial-positive-spin2}), we thus conclude

\begin{lemma} Fix $M>0$, $0<|a|\leq M$ and $s\in\{\frac52,3\}$.
There exist real algebraically special frequencies, i.e.\ real $(\omega,m,l)$, such that $\mathfrak{C}_s(a,M,\omega,m,l)=0$.
\end{lemma}

\section{On non-superradiant amplification for \texorpdfstring{$s>2$}{high spin} fields}
\label{sec:nonsuperradiant-amplification}

In this section, we consider the implications of the properties of the radial Teukolsky--Starobinsky constants on the energy associated to the Teukolsky radial ODE~\eqref{eq:radial-ODE-alpha}. 

\subsection{Energy for the Teukolsky equation}

We begin by discussing a notion of energy compatible with \eqref{eq:Teukolsky-equation-intro} for any $s\in\frac12\mathbb{Z}_{\geq 0}$ with the stationary Kerr Killing field fails to produce a conservation law unless $s= 0$. However, as $(r^2+a^2)^{\frac12}\Delta^{\pm \frac{s}{2}}\smlambda{\upalpha}{\pm s}$ satisfy complex conjugate ODEs, see \eqref{eq:radial-ODE-alpha}, the Wronskian
\begin{align*}
\lp(\Delta\frac{d}{dr}(\Delta^s\smlambda{\upalpha}{+s})\Delta^{-s}\overline{\smlambda{\upalpha}{-s}}-\Delta\frac{d}{dr}\lp(\overline{\smlambda{\upalpha}{-s}}\rp)\smlambda{\upalpha}{+s}\rp)
\end{align*}
is independent of $r$, and hence is conserved. From the Teukolsky--Starobinsky identities, given a solution $\smlambda{\upalpha}{+s}$ to \eqref{eq:radial-ODE-alpha} with spin $+s$, we may generate a $\smlambda{\upalpha}{-s}$ which solves \eqref{eq:radial-ODE-alpha} with spin $-s$, and vice-versa, to plug into this conservation law. We thus obtain:

\begin{lemma}[TS energy identity] \label{lemma:TS-energy-identity} Fix $M>0$, $|a|\leq M$, $s\in\frac12\mathbb{Z}_{\geq 0}$, $\omega\in\mathbb{R}\backslash\{0,m\upomega_+\}$ and $\Lambda\in\mathbb{R}$. Suppose there exists a radial Teukolsky--Starobinsky constant $\mathfrak{C}_s(a,M,\omega,m,\Lambda)$ as given in Definition~\ref{def:TS-radial-constant}. 

Let $\smlambda{\upalpha}{\pm s}$ be a solution to \eqref{eq:radial-ODE-alpha}, and  let its decomposition \eqref{eq:radial-ODE-decomposition} be characterized by 
\begin{align*}
\swei{a}{\pm s}_{\mc{H}^-}=0\,, \qquad \swei{\tilde a}{\pm s}_{\mc{H}^+ }\equiv \swei{a}{\pm s}_{\mc{H}^+ }\begin{dcases}
(2Mr_+)^{1/2}\lp\{\begin{array}{ll}
1\,, & |a|=M\\
(r_+-r_-)^{\pm s}\,, & |a|<M
\end{array}\rp\}  &\text{~if $s$ integer}\\
\lp\{\begin{array}{ll}
(2Mr_+)^{-1/2}\,, & |a|=M\\
(r_+-r_-)^{\pm s-1/2}\,, & |a|<M
\end{array}\rp\}  &\text{~if $s$ half-integer}
\end{dcases}
\,. 
\end{align*}
Then, it satisfies the energy identity 
\begin{align*}
1 = \swei{\mathfrak{R}}{\pm s}(a,M,\omega,m,\Lambda)+
 \swei{\mathfrak{T}}{\pm s}(a,M,\omega,m,\Lambda)\,,  
\end{align*}
where $\swei{\mathfrak{R}}{\pm s}(a,M,\omega,m,\Lambda)$  and $\swei{\mathfrak{T}}{\pm s}(a,M,\omega,m,\Lambda)$ are called the reflection and transmission coefficients, respectively, and are given as follows. If $s$ is an integer,
\begin{align*}
\begin{split}
\swei{\mathfrak{T}}{-s}&=\frac{\omega-m\upomega_+}{\omega}\frac{\mathfrak C_s^{(10)}}{(2\omega)^{2s}}\frac{\lp|\swei{\tilde a}{- s}_{\mc{H}^+ }\rp|^2}{\lp|\swei{a}{- s}_{\mc{I}^- }\rp|^2}\,, \quad
\swei{\mathfrak{R}}{-s}=\frac{\mathfrak C_s}{(2\omega)^{4s}}\frac{\lp|\swei{ a}{- s}_{\mc{I}^+ }\rp|^2}{\lp|\swei{a}{- s}_{\mc{I}^- }\rp|^2}\,;\\
\text{if further $\mathfrak{C}_s\neq 0$\,,} \quad \swei{\mathfrak{T}}{+s}&=\frac{\omega-m\upomega_+}{\omega}\frac{(2\omega)^{2s}}{\mathfrak C_s^{(9)}}\frac{\lp|\swei{\tilde a}{+ s}_{\mc{H}^+ }\rp|^2}{\lp|\swei{a}{+ s}_{\mc{I}^- }\rp|^2}\,,\quad
\swei{\mathfrak{R}}{+s}=\frac{(2\omega)^{4s}}{\mathfrak C_s}\frac{\lp|\swei{ a}{+ s}_{\mc{I}^+ }\rp|^2}{\lp|\swei{a}{+ s}_{\mc{I}^- }\rp|^2}\,.
\end{split} 
\end{align*}
If $s$ is half-integer, 
\begin{align*}
\begin{split}
\swei{\mathfrak{T}}{-s}&=\frac{\mathfrak C_s^{(10)}}{(2\omega)^{2s}}\frac{\lp|\swei{\tilde a}{- s}_{\mc{H}^+ }\rp|^2}{\lp|\swei{a}{- s}_{\mc{I}^- }\rp|^2}\,, \quad
\swei{\mathfrak{R}}{-s}=\frac{\mathfrak C_s}{(2\omega)^{4s}}\frac{\lp|\swei{ a}{- s}_{\mc{I}^+ }\rp|^2}{\lp|\swei{a}{- s}_{\mc{I}^- }\rp|^2}\,;\\
\text{if further $\mathfrak{C}_s\neq 0$\,,} \quad\swei{\mathfrak{T}}{+s}&=\frac{(2\omega)^{2s}}{\mathfrak C_s^{(9)}}\frac{\lp|\swei{\tilde a}{+ s}_{\mc{H}^+ }\rp|^2}{\lp|\swei{a}{+ s}_{\mc{I}^- }\rp|^2}\,,\quad
\swei{\mathfrak{R}}{+s}=\frac{(2\omega)^{4s}}{\mathfrak C_s}\frac{\lp|\swei{ a}{+ s}_{\mc{I}^+ }\rp|^2}{\lp|\swei{a}{+ s}_{\mc{I}^- }\rp|^2}\,.
\end{split} 
\end{align*}
Here, we take $\mathfrak{C}_s^{(9)}=\mathfrak{C}_s^{(10)}=1$ when $s=0$, $\mathfrak{C}_s^{(9)}=1$ and $\mathfrak{C}_s^{(10)}=\lp[4Mr_+(\omega-m\upomega_+)\rp]^2+(r_+-r_-)^2/4$ if $s=\pm 1/2$, and otherwise, using the shorthand notation $\mathfrak{C}_{s,j}:=\lp[4Mr_+(\omega-m\upomega_+)\rp]^2+(s-j)^2(r_+-r_-)^2$,
\begin{align*}
\mathfrak{C}_s^{(9)}&=
\begin{dcases}
\displaystyle\prod_{j=1}^{|s|} \mathfrak{C}_{s,j} &\text{if~} s\in\mathbb{Z}\\
\displaystyle\prod_{j=1}^{|s|-1/2} \mathfrak{C}_{s,j} &\text{if~} s\in\lp(\frac12\mathbb{Z}\rp)\backslash\mathbb{Z}
\end{dcases}\,,\qquad
\mathfrak{C}_s^{(10)}=
\begin{dcases}
\displaystyle\prod_{j=0}^{|s|-1} \mathfrak{C}_{s,j} &\text{if~} s\in\mathbb{Z}\\
\displaystyle\prod_{j=0}^{|s|-1/2} \mathfrak{C}_{s,j} &\text{if~} s\in\lp(\frac12\mathbb{Z}\rp)\backslash\mathbb{Z}
\end{dcases}\,.
\end{align*}
\end{lemma}

\begin{proof} The proof is sketched in the paragraph above, but we encourage the reader to see \cite[Section IIB]{Andersson2017} and \cite{SRTdC2020} for details.
\end{proof}

\begin{remark} The notion of energy put forth in Lemma~\ref{lemma:TS-energy-identity} is consistent with previous literature: it matches that introduced in \cite{Teukolsky1974} for $s=1,2$,  \cite{Unruh1973} for $s=1/2$ and \cite{TorresdelCastillo1990} for $s=3/2$.
\end{remark}

If $\swei{\mathfrak T}{\pm s}<0$ and $\swei{\mathfrak R}{\pm s}>0$, there is \textit{amplication} in the energy reflected to future null infinity. As
\begin{align*}
\swei{\mathfrak T}{\pm s}(a,M,\omega,m,\Lambda)<0\Leftrightarrow \omega(\omega-m\upomega_+)<0\,,\quad s\in\mathbb{Z}_{\geq 0}\,,
\end{align*}
we refer to this as superradiant amplification. On the other hand, if $\swei{\mathfrak T}{\pm s}>0$ and $\swei{\mathfrak R}{\pm s}<0$, there is \textit{amplication} in the energy transmitted into the future event horizon. As
\begin{align*}
\swei{\mathfrak R}{\pm s}(a,M,\omega,m,\Lambda)<0\Leftrightarrow \mathfrak{C}_s(a,M,\omega,m,\Lambda)<0\,,
\end{align*}
we refer to this as non-superradiant amplification. By  Lemma~\ref{lemma:TS-radial-constant-nonnegative} and Lemma~\ref{lemma:TS-radial-constant-negative}, if we constrain $\Lambda$ to be a spin-weighted spheoidal eigenvalue, $\Lambda=\bm\Lambda_{sml}^{(a\omega)}$ for some $l$, we see that non-superradiant amplification  occurs for some $(\omega,m,l)$ when $s>2$, see Figure~\ref{fig:T3R3}. 

\begin{figure}[t]
\centering   
\begin{subfigure}{.5\textwidth}
  \centering
  \includegraphics[width=.9\textwidth]{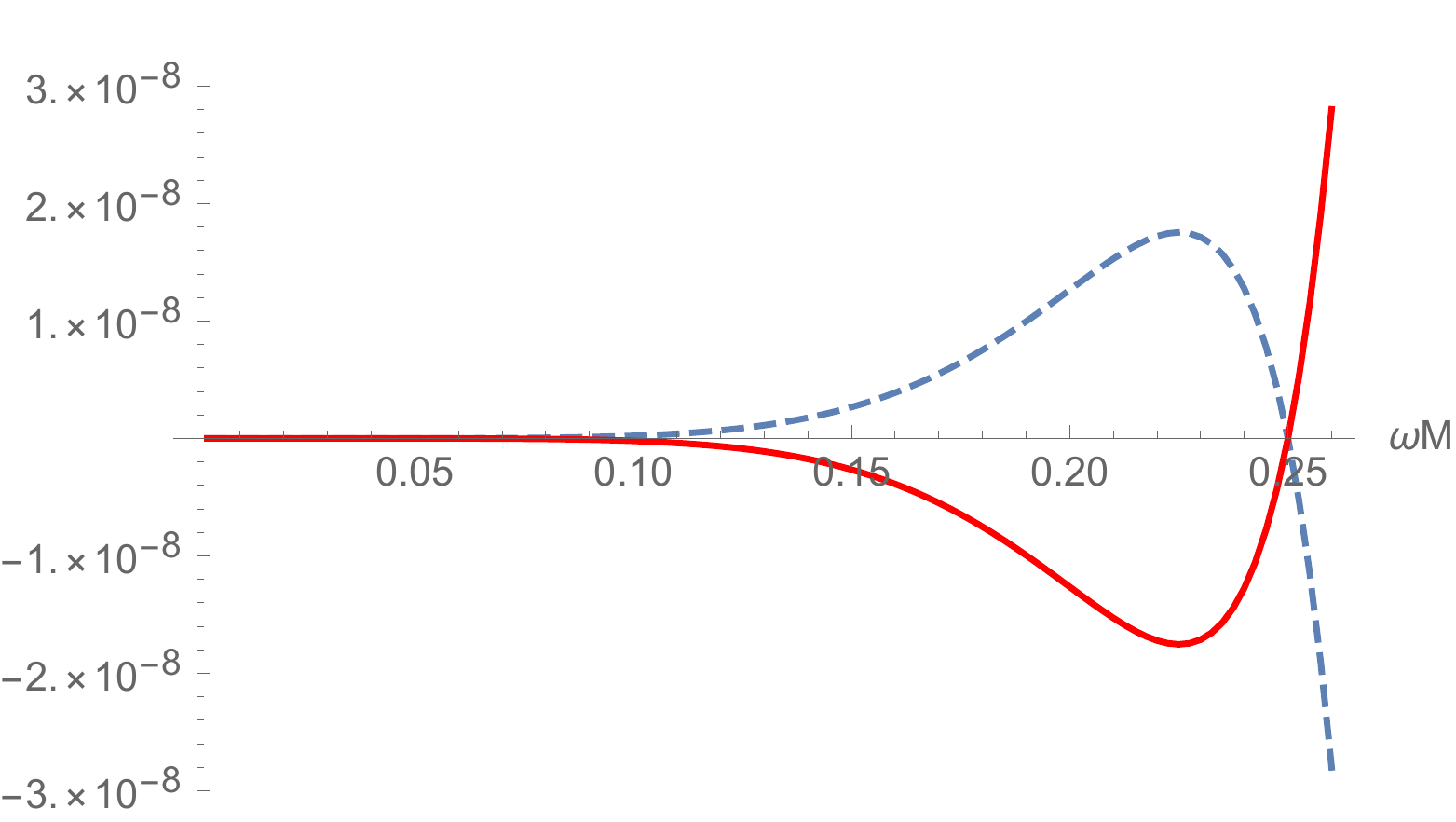}
  \caption{}
  \label{fig:T3}
\end{subfigure}%
\begin{subfigure}{.5\textwidth}
  \centering
  \includegraphics[width=.9\textwidth]{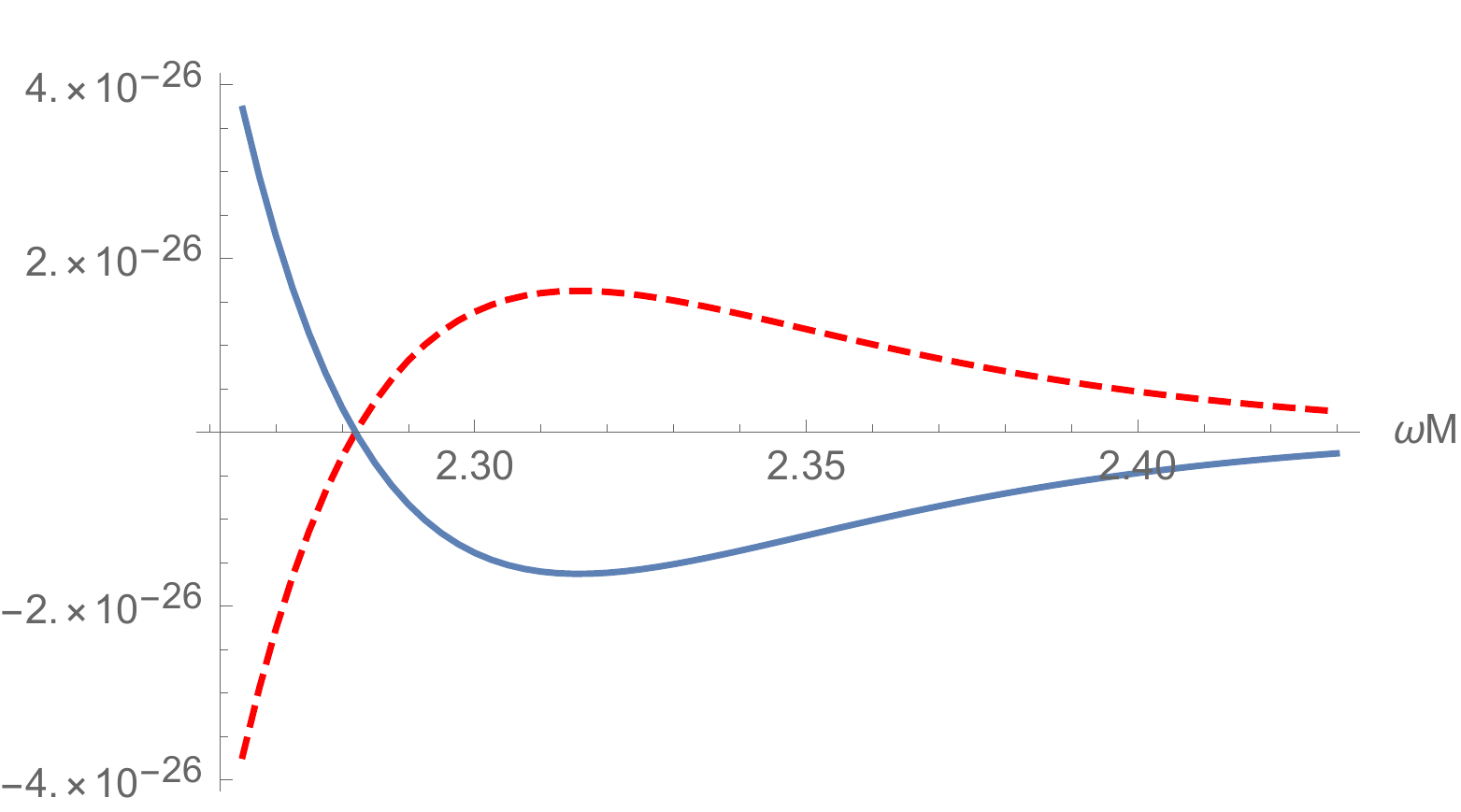}
  \caption{}
  \label{fig:R3}
\end{subfigure}
\caption{Plot showing the behavior of $\swei{\mathfrak T}{-3}(a, M, \omega,m,l)$ and $\swei{\mathfrak R}{-3}(a,M,\omega,m,l)$ for $a=8M/10$, $m=1$ and $l=3$, as functions of $M\omega$. Plot (a) captures superradiant amplification, with the solid red line corresponding to $\swei{\mathfrak T}{-3}$, and the dashed blue line denoting $\swei{\mathfrak R}{-3}-1$. Plot (b) captures non-superradiant amplification, with the solid blue line corresponding to $\swei{\mathfrak R}{-3}$, and the dashed red line denoting $\swei{\mathfrak T}{-3}-1$.}
\label{fig:T3R3}
\end{figure}

\subsection{Energy for the system of linearized Maxwell or Einstein equations}

In the previous section, we considered a scattering problem under the evolution equation \eqref{eq:Teukolsky-equation-intro} alone. Consequently, the notion of energy in Lemma~\ref{lemma:TS-energy-identity} involves a single spin sign. If, for $s=1,2$, one considers a scattering problem under the entire system of linearized Maxwell or Einstein equations, respectively, then the natural notion of energy involves both spin signs. We extend this reasoning to other spins:

\begin{lemma}[TS energy identity under TS correspondence] \label{lemma:TS-energy-identity-with-TS-correspondence} Fix $M>0$, $|a|\leq M$, $s\in\frac12\mathbb{Z}_{\geq 0}$, $\omega\in\mathbb{R}\backslash\{0,m\upomega_+\}$ and $\Lambda\in\mathbb{R}$. Suppose there exists a radial Teukolsky--Starobinsky constant $\mathfrak{C}_s(a,M,\omega,m,\Lambda)$ as given in Definition~\ref{def:TS-radial-constant}. Further assume that the frequencies are such that $\mathfrak{C}_s(a,M,\omega,m,\Lambda)\neq 0$.

Let $\smlambda{\upalpha}{\pm s}$ be solutions to \eqref{eq:radial-ODE-alpha} which are related to each other by the radial Teukolsky--Starobinsky identities of Proposition~\ref{prop:TS-radial}. Assume their decompositions \eqref{eq:radial-ODE-decomposition} to be characterized by
\begin{align*}
\swei{a}{- s}_{\mc{H}^-}=0\,, \qquad \swei{\tilde a}{+ s}_{\mc{H}^+ }\equiv \swei{a}{+ s}_{\mc{H}^+ }\begin{dcases}
(2Mr_+)^{1/2}\lp\{\begin{array}{ll}
1\,, & |a|=M\\
(r_+-r_-)^{s}\,, & |a|<M
\end{array}\rp\}  &\text{~if $s$ integer}\\
\lp\{\begin{array}{ll}
(2Mr_+)^{-1/2}\,, & |a|=M\\
(r_+-r_-)^{+ s-1/2}\,, & |a|<M
\end{array}\rp\}  &\text{~if $s$ half-integer}
\end{dcases}
\,. 
\end{align*}
Then, one has  the energy identity 
\begin{align*}
1 = \mathfrak{R}_s(a,M,\omega,m,\Lambda)+\mathfrak{T}_s(a,M,\omega,m,\Lambda)
\,,  
\end{align*}
where the reflection and transmission coefficients, $\mathfrak{R}_s(a,M,\omega,m,\Lambda)$ and $\mathfrak{T}_s(a,M,\omega,m,\Lambda)$ respectively, are given as follows:
\begin{align*}
\begin{split}
\mathfrak{R}_{s}\equiv \frac{\lp|\swei{ a}{- s}_{\mc{I}^+ }\rp|^2}{\lp|\swei{ a}{+ s}_{\mc{I}^- }\rp|^2}\,; \quad \text{$s$ integer,}\,\,\, \mathfrak{T}_{s}\equiv\frac{\omega-m\upomega_+}{\omega}\frac{(2\omega)^{2s}}{\mathfrak C_s^{(9)}}\frac{\lp|\swei{\tilde a}{+ s}_{\mc{H}^+ }\rp|^2}{\lp|\swei{ a}{+ s}_{\mc{I}^- }\rp|^2}\,; \quad \text{$s$ half-integer,}\,\,\,\mathfrak{T}_{s}\equiv\frac{(2\omega)^{2s}}{\mathfrak C_s^{(9)}}\frac{\lp|\swei{\tilde a}{+ s}_{\mc{H}^+ }\rp|^2}{\lp|\swei{ a}{+ s}_{\mc{I}^- }\rp|^2}\,;
\end{split} 
\end{align*}
where we take $\mathfrak{C}_s^{(9)}$  to be the same as in Lemma~\ref{lemma:TS-energy-identity}.
\end{lemma}

\begin{remark} The notion of energy considered in Lemma~\ref{lemma:TS-energy-identity-with-TS-correspondence} matches that of the recent \cite[Section 1.3.4]{Masaood2020} on scattering under the linearized Einstein vacuum equations around $a=0$ in Kerr.  Indeed, when considering this system, the two gauge-invariant curvature quantities satisfy the Teukolsky Master equation \eqref{eq:Teukolsky-equation-intro} with spin $\pm 2$ together with a physical space version of the radial Teukolsky--Starobinsky identities of Proposition~\ref{prop:TS-radial}, see \cite[Equations 1.5 and 1.6]{Masaood2020}. A similar situation arises when one considers the linearized Maxwell equations, as one may readily deduce from \cite[Equations 3.11--3.16]{Pasqualotto2016}.
\end{remark}

We note that it is only under the notion of energy in Lemma~\ref{lemma:TS-energy-identity-with-TS-correspondence}, seldom encountered in the physics literature, that amplification does not occur for non-superradiant frequencies for fields of any spin.

\printbibliography

\end{document}